\documentclass[12pt,a4paper]{iopart}
\usepackage[utf8]{inputenc}
\usepackage{lineno}

\usepackage{amssymb,bm,url,mathrsfs,cite}

\expandafter\let\csname equation*\endcsname\relax
\expandafter\let\csname endequation*\endcsname\relax
\usepackage{amsmath}
\usepackage{algorithmic}
\usepackage{algorithm}
\usepackage{amsthm}
\usepackage{hyperref}
\usepackage{makecell}
\usepackage{tablefootnote}
\usepackage{comment}
\usepackage{color}
\usepackage{graphics,graphicx}
\usepackage{mathabx}
\usepackage{subcaption,caption}
\usepackage[capitalise]{cleveref}
\usepackage{bbm}
\usepackage{epstopdf}
\usepackage{booktabs}
\hypersetup{
	colorlinks=true,
	linkcolor=blue,
	citecolor=blue,
	filecolor=magenta,
	urlcolor=cyan,
}

\graphicspath{{figure/}}

\definecolor{OliveGreen}{rgb}{0,0.6,0}

\def\A{\boldsymbol{A}}

\def\Y{\boldsymbol{Y}}
\def\I{\boldsymbol{I}}
\def\E{\boldsymbol{E}}
\def\M{\boldsymbol{M}}

\def\indicator{\mathbbm{1}}
\def\wt{\widetilde}

\def\hat{\widehat}

\def\no{\notag}
\def\tilde{\widetilde}

\def\supp{\text{supp}}
\def\var{\text{Var}}
\def\ybar{\widebar{\Y}}
\def\xmax{x_{\max}}
\def\sign{\textsf{sign}}
\def\dist{\mathrm{dist}}
\def\stablerank{\textsf{sr}}

\newcommand\stablesparsity[1]{\bar{s}_{#1}}
\newcommand\fro[1]{\| #1 \|_{\rm{F}}}

\newcommand\op[1]{\|#1\|}
\newcommand\lzero[1]{\|#1\|_{\ell_0}}
\newcommand\ltwo[1]{\|#1\|_{\ell_2}}
\newcommand\linf[1]{\|#1\|_{\ell_{\infty}}}
\newcommand\inp[2]{\langle #1, #2 \rangle}

\newcommand{\blue}[1]{{\color{black}#1}}

\newcommand{\RN}[1]{%
	\textup{\uppercase\expandafter{\romannumeral#1}}%
}

\def\calE{{\mathcal E}}

\def\calN{{\mathcal N}}
\def\calO{{\mathcal O}}

\def\calS{{\mathcal S}}
\def\calT{{\mathcal T}}

\def\EE{{\mathbb E}}

\def\MM{{\mathbb M}}

\def\PP{{\mathbb P}}

\def\RR{{\mathbb R}}
\def\SS{{\mathbb S}}
\def\TT{{\mathbb T}}

\def\a{{\boldsymbol a}}

\def\e{{\boldsymbol e}}

\def\u{{\boldsymbol u}}
\def\v{{\boldsymbol v}}
\def\w{{\boldsymbol w}}
\def\x{{\boldsymbol x}}
\def\y{{\boldsymbol y}}
\def\z{{\boldsymbol z}}

\newtheorem{theorem}{Theorem}

\newtheorem{lemma}{Lemma}
\newtheorem{corollary}{Corollary}
\newtheorem{proposition}{Proposition}

\let\oldequation\equation
\let\oldendequation\endequation

\renewenvironment{equation}
  {\linenomathNonumbers\oldequation}
  {\oldendequation\endlinenomath}

\let\oldalign\align
\let\oldendalign\endalign

\renewenvironment{align}
  {\linenomathNonumbers\oldalign}
  {\oldendalign\endlinenomath}

\begin{document}

\title[Sparse Phase Retrieval via Truncated Power Method]{Provable Sample-Efficient Sparse Phase Retrieval Initialized by Truncated Power Method}

\author{Jian-Feng Cai, Jingyang Li, and Juntao You}
\address{Department of Mathematics, The Hong Kong University of Science and Technology, Clear Water Bay, Kowloon, Hong Kong SAR, China}

\ead{{jfcai@ust.hk}, {jlieb@connect.ust.hk}, {jyouab@connect.ust.hk}}

\vspace{10pt}
\begin{indented}
\item[]April 2023
\end{indented}

\begin{abstract}
We study the sparse phase retrieval problem, recovering an  $s$-sparse length-$n$ signal from $m$ magnitude-only measurements.  Two-stage non-convex approaches have drawn much attention in recent studies. Despite non-convexity, many two-stage algorithms provably converge to the underlying solution linearly when appropriately initialized. However, in terms of sample complexity, the bottleneck of those algorithms with Gaussian random measurements often comes from the initialization stage. Although the refinement stage usually needs only $m=\Omega(s\log n )$ measurements, the widely used spectral initialization in the initialization stage requires $m=\Omega(s^2\log n )$ measurements to produce a desired initial guess, which causes the total sample complexity order-wisely more than necessary. To reduce the number of measurements, we propose a truncated power method to replace the spectral initialization for non-convex sparse phase retrieval algorithms. We prove that $m=\Omega(\bar{s} s\log n )$ measurements, where $\bar{s}$ is the stable sparsity of the underlying signal, are sufficient to produce a desired initial guess. When the underlying signal contains only very few significant components, the sample complexity of the proposed algorithm is $m=\Omega(s\log n )$ and optimal. Numerical experiments illustrate that the proposed method is more sample-efficient than state-of-the-art algorithms.

\end{abstract}

\section{Introduction}

In this paper, we consider the real-valued phase retrieval problem, which aims at recovering a high-dimensional signal from a system of phaseless measurements. That is, the task is to recover $\x\in\RR^n$ from the system
\begin{align}\label{problem:PR}
y_i = |\inp{\a_i}{\x}|,~~~i = 1,\ldots,m,
\end{align}
where $\{\a_i\}_{i=1}^{m}\subset \mathbb{R}^n$ are the sensing vectors and $\{y_i\}_{i=1}^{m}\subset\mathbb{R}_+$ are the observations. The phase retrieval problem has arisen in various applications, for instances, X-ray crystallography \cite{harrison1993phase}, optics \cite{walther1963question}, microscopy \cite{miao2008extending}, and others \cite{fienup1982phase,shechtman2015phase}. In those applications, it is easier to record the intensity of the measurements due to hardware limitations, which leads to the phase retrieval problem \eqref{problem:PR}.\par

When there is no prior assumption on $\x$, there may be infinitely many possible solutions to \eqref{problem:PR} if $m<n$.
To ensure the well-posedness of the problem, an \emph{oversampling} (i.e., $m>n$) technique is applied. It has been shown that $m\ge 2n-1$ measurements are necessary and sufficient for a unique recovery (up to a global phase) with generic real sensing vectors \cite{balan2006signal}. Moreover, recent studies further indicate that the oversampling is required for both practical algorithms \cite{candes2015phase,waldspurger2015phase,goldstein2016phasemax,hand2018elementary,bahmani2017flexible,netrapalli2013phase,candes2015phase1, chen2017solving,zhang2016reshaped,wang2016solving,cai2018solving} and global landscape analysis \cite{sun2018geometric,li2020Toward,cai2021global1,cai2022global2,cai2022solving,cai2022nearly} to guarantee a successful recovery.

Meanwhile, the over-complete measurements $m=\Omega(n)$ naturally result in a huge sampling and computation cost in high-dimensional applications. Recently, it has been a great interest to further reduce the required number of measurements (or sample size) for phase retrieval problem \cite{wang2014phase,phase15,cai2016optimal,salehi2018learning,bahmani2015efficient,iwen2017robust}. To recover the underlying signal $\x\in\RR^n$ with underdetermined measurements $m<n$, a common approach is to exploit the latent structure of the signal $\x$. In many real applications related to signal/imaging processing, it is well-known that the underlying signal  $\x$ is usually (approximately) sparse in a transformed domain \cite{stephane1999wavelet}. With sparsity priors, one then has the so-called sparse phase retrieval problem, which is to find a sparse signal $\x\in \mathbb{R}^n$ from the system
\begin{align}\label{problem:sPR}
y_i = |\inp{\a_i}{\x}|,~i = 1,\ldots,m,\quad \text{subject to} \quad \lzero{\x}  \le s,
\end{align}
\blue{where $\lzero{\x}$ is the number of non-zero entries of $\x$} and $s\ll n$ is the sparsity. It has been shown that the problem \eqref{problem:sPR} admits a unique solution (up to a global phase) with only $m=2s$ real generic measurements \cite{wang2014phase}. One of the remaining challenges of the sparse phase retrieval problem \eqref{problem:sPR} is introducing practical algorithms with (near) optimal sample complexity.

\subsection{Related work and contributions}

\paragraph{Related work.}

Despite the non-convexity and NP-hardness, a number of practical algorithms have been introduced for \eqref{problem:sPR} with theoretical guarantees under the assumption of i.i.d. Gaussian sensing vectors $\{\a_i\}_{i=1}^m$. Those non-convex methods are usually divided into two stages, namely, the initialization stage and the local refinement stage. Provided an initial guess that is sufficiently close to the underlying signal, non-convex algorithms including SPARTA~\cite{wang2017sparse}, CoPRAM~\cite{jagatap2019sample}, thresholding/projected Wirtinger flow~\cite{cai2016optimal,soltanolkotabi2019structured}, SAM~\cite{cai2022sample} and HTP~\cite{cai2022sparse} are guaranteed to converge at least linearly to the ground truth with the near-optimal sample complexity $\Omega(s\log n)$.
Moreover, algorithms like SAM and HTP are guaranteed to give an exact recovery of the underlying signal in only a few iterations, implying that the local refinement stage can be efficiently implemented. On the other hand, the popular spectral initialization and its variations require $m = \Omega(s^2\log n)$ samples to obtain an accurate initial guess \cite{phase15,cai2016optimal,wang2017sparse,jagatap2019sample}. Therefore, for non-convex two-stage methods, the bottleneck in sample complexity comes from the initialization stage. Consequently, it is possible to break the sample complexity barriers as long as the methods for initialization can be improved.

Alternatively, a Hadamard Wirtinger flow (HWF) method has been introduced in \cite{wu2021hadamard} as a different strategy to obtain the initial guess, where an implicit regularization is used. The sample complexity of HWF is improved to $\Omega(\stablesparsity{\x}s\log n )$ under the assumption that $x_{\min} = \Omega(\frac{1}{\sqrt{s}}\ltwo{\x})$, where $\|\cdot\|_{\ell_2}$ denotes the $\ell_2$-norm for vectors, $\stablesparsity{\x}:=\frac{\ltwo{\x}^2}{\linf{\x}^2}$ is the stable sparsity of $\x$, and $x_{\min}$ is the minimum nonzero entry in absolute value of $\x$. However, the assumption on $x_{\min}$ requires a very small dynamic range of the nonzero entries of $\x$, which is too stringent to be satisfied in many practical situations. For example, the assumption on $x_{\min}$ fails to hold for signals whose components decay fast, e.g., in the form of power series \cite{chen2015exact}. It is also worth mentioning that if the sampling vectors $\{\bm{a}_i\}_{i=1}^{m}$ are not standard Gaussian, a two-stage sampling scheme requiring only $\Omega(s\log(n/s))$ samples has been introduced in \cite{iwen2017robust}, where The first stage is for sparse compressed sensing and the second stage is for phase recovery. Nevertheless, for Gaussian random measurements, there is still a statistical-to-computational gap to overcome.

\paragraph{Our contributions.} This paper focuses on the standard Gaussian model and aims to reduce the sampling complexity of provably non-convex sparse phase retrieval algorithms by improving the initialization stage. In this work, we introduce two novel initialization algorithms, namely, a modified spectral initialization method and a novel efficient truncated power method for sparse phase retrieval problems.
The proposed algorithms can be applied as initialization methods for all the aforementioned non-convex algorithms. Without any assumption on the underlying signal, we prove that  $\Omega(\stablesparsity{\x}s\log n )$ measurements are sufficient for the modified spectral initialization method and truncated power method to obtain a desired initial guess.
Since $\x$ is known to be $s$-sparse, we see from the definition of $\stablesparsity{\x}$ that $1\leq\stablesparsity{\x}\leq s$. So, when $\x$ has only very few significant entries and thus $\stablesparsity{\x}=\calO(1)$, our algorithm requires only $\Omega(s\log(n))$ samples, which is nearly optimal up to a logarithmic factor. Consequently, the total sample complexity of many non-convex sparse phase retrieval algorithms can be reduced to $m=\Omega(s\log n)$ as long as $\stablesparsity{\x}=\calO(1)$.
Moreover, the truncated power method has a better theoretical sample complexity than HWF to obtain a $\delta$-close initial guess $\hat\x$ (i.e.,  $\dist(\x,\hat\x):=\min\{\ltwo{\x-\hat\x},\ltwo{\x+\hat\x}\}\leq \delta\ltwo{\x}$,). See \cref{tab:samplecom} for a comparison of theoretical results of different initialization methods. We will demonstrate by experimental results the superior performance of the proposed methods over other initialization methods.

\begin{table}[t]
\centering
\caption{\label{tab:samplecom}Comparison of our proposed initialization methods with state-of-the-art initialization methods for sparse phase retrieval \blue{under the regime $s= \Theta(n^{\sigma})$ for some $\sigma\in(0,1)$}. The sample complexity is to guarantee a $\delta$-close initial guess. Here $\stablesparsity{\x}:= \ltwo{\x}^2/\linf{\x}^2\in[1,s]$ is  the stable sparsity of $\x$.}\label{tab:samplecomlexity}
\begin{tabular}{c|c|c}
\toprule
 Initialization methods  &Sample complexity $m$  &Assumptions   \\
\midrule
Spectral method~\cite{phase15,cai2016optimal,wang2017sparse,jagatap2019sample}   & $\Omega(\delta^{-2}s^2\log n)$   &none~\cite{jagatap2019sample}   \\
\hline
HWF~\cite{wu2021hadamard} & $\Omega(\delta^{-2}\stablesparsity{\x}s\log n )$   & $x_{\min} \ge \frac{c}{\sqrt{s}}\ltwo{\x}$   \\
\hline
Modified spectral method (this paper) & \blue{$\Omega(\delta^{-2}\stablesparsity{\x}s\log n )$}   & none  \\
\hline
Truncated power method (this paper)
 & \blue{$\Omega(\max\{\delta^{-2},\stablesparsity{\x}\}s\log n)$}   & none \\
\bottomrule
\end{tabular}
\end{table}

\subsection{Notations and outline}
We use boldface capital and lower-case for matrices and vectors respectively (e.g., $\A$ and $\x$). We use square brackets with subscripts (e.g., $[\A]_{ij},[\x]_i$) or Roman letters (e.g., $A_{ij},x_i$) to represent the corresponding entries of matrices and vectors.
We use $\A_{\calT}$ or $[\A]_{\calT}$ to denote the sub-matrix of $\A$ whose columns and rows are indexed by $\calT$ and $\x_{\calT}$ the sub-vector of $\x$.
We denote $\|\cdot\|_{\ell_p}$ the $\ell_p$-norm for vectors, i.e., $\|\x\|_{\ell_p} = (\sum_{i=1}^n|x_i|^p)^{\frac{1}{p}}$ for all $0<p<+\infty$. The notation $\lzero{\x}$ stands for the number of non-zero entries of $\x$, and $\linf{\x}$ is the largest entry of $\x$ in absolute value.
Denote $\fro{\cdot}$ the Frobenius norm and $\op{\cdot}$ the operator norm of matrices.
For any $1\leq \alpha\leq 2$, the Orlicz norm of a random variable is defined by  $\|X\|_{\psi_{\alpha}}:=\min\{K>0: \EE\exp\{|X/K|^{\alpha}\}\leq 2\}$.
Throughout, we use $\Omega(\cdot)$ for asymptotic lower bounds, $\calO(\cdot)$ for upper bounds, and $\Theta(\cdot)$ for both lower and upper bounds, i.e., $m=\Omega(f(s))$ means $|m/f(s)|\ge C$ for some universal constant $C>0$ when $s$ is sufficiently large and $m=\calO(f(s))$ means $|m/f(s)|\leq C'$ for some universal constant $C'>0$ when $s$ is sufficiently large, and $m=\Theta(f(s))$ means $C''\leq|m/f(s)|\leq C'''$ for some universal constants $C'',C'''>0$ when $s$ is sufficiently large.

Let $\SS^n = \{\A\in\RR^{n\times n}: \A^T = \A\}$ denote the set of symmetric matrices. For any $\A\in\SS^n$, we denote its eigenvalues by $\lambda_{\min}(\A) = \lambda_n(\A)\leq \cdots\leq \lambda_1(\A) = \lambda_{\max}(\A)$, and we use $\rho(\A)$ to denote the spectral norm of $\A$, i.e., $\rho(\A) = \max\{|\lambda_1(\A)|,|\lambda_{n}(\A)|\}$. Also, we define the largest and smallest $s$-sparse eigenvalue by
\begin{align*}
\lambda_{\max}(\A,s) = \max_{\w\in\RR^n,\lzero{\w}\leq s, \ltwo{\w} \leq 1}\w^T\A\w,\quad \lambda_{\min}(\A,s) = \min_{\w\in\RR^n,\lzero{\w}\leq s, \ltwo{\w} \leq 1}\w^T\A\w,
\end{align*}
respectively, and
\begin{equation}\label{eq:defrho}
\rho(\A,s) := \max\{|\lambda_{\min}(\A,s)|,|\lambda_{\max}(\A,s)|\}.
\end{equation}
We denote $\xmax = \linf{\x}$, $\bm{y}=[y_1,y_2,\cdots,y_m]^T$, and $\{\bm{e}_j\}_{j=1}^n$ is the standard basis vectors of $\mathbb{R}^n$.

The rest of the paper is organized as follows. We introduce the proposed algorithms in \cref{section:algorithm}. Our main theoretical results are presented in \cref{section:theory}, and their proofs are given in \cref{section:proofs}. In \cref{section:experiments}, we provide some numerical experiments to demonstrate the performance of the proposed algorithms.

\section{Algorithms}\label{section:algorithm}
For non-convex algorithms, the desired initial guess should be sufficiently close to the ground truth (up to a global phase). In this section, we propose two algorithms that provide accurate initialization for non-convex sparse phase retrieval, using as few samples as possible.
The first algorithm is a truncated power method \cite{yuan2013truncated} in \cref{section:tpm}.
The second algorithm is a modified spectral initialization in \cref{section:ossi}, which is a modification of the spectral method \cite{wang2017sparse,jagatap2019sample} and also a special case of truncated power method without the refinement step. For completeness, we first briefly review spectral methods for phase retrieval in \cref{sec:specinitgpr} and sparse phase retrieval in \cref{sec:specinit}.

\subsection{Spectral initialization for general phase retrieval}\label{sec:specinitgpr}
The standard spectral method is designed to initialize general non-convex phase retrieval algorithms without the sparsity assumption. It constructs the following matrix
\begin{equation}\label{eq:defY}
	\Y = \frac{1}{m}\sum_{i=1}^m y_i^2\a_i\a_i^T.
\end{equation}
%whose expectation is given by $\EE\Y = \I + 2\x\x^T$,
It is easy to check that, when the measurement vectors $\a_i$, $i=1,\ldots,m$, are i.i.d. random Gaussian, the expectation of $\Y$ is given by
\begin{equation}\label{eq:EY}
    \EE \Y=\ltwo{\x}^2\I+2\x\x^T.
\end{equation}
So, any principal eigenvector of $\EE\Y$ is a multiple of $\x$. Therefore, when $\Y$ has a high concentration around its expectation, a principle eigenvector of $\Y$ with a suitable length provides a good approximation to the underlying signal $\x$, and hence a good initialization for phase retrieval algorithms.
To improve the accuracy and reduce the sample complexity, instead of $\Y$ itself, one may consider its truncated version
\begin{equation}\label{eq:defYbar0}
	\ybar_0 = \frac{1}{m}\sum_{i=1}^m y_i^2\a_i\a_i^T\indicator_{\{l\leq (y_i/\ltwo{\x})\leq u\}},
\end{equation}
where $0\leq l<1\leq u$ are the truncation parameters. It can still be shown (see Lemma \ref{lemma:expectation}) that any principle eigenvector of $\EE\ybar_0$ is a multiple of $\x$. However, the quantity $\ltwo{\x}$ is not available in the observed data, and one may use
\begin{equation}\label{eq:defnu}
\nu = \left(\frac{1}{m}\sum_{i=1}^m y_i^2\right)^{1/2}
\end{equation}
as an estimation of $\ltwo{\x}$ (see \cref{lemma:nu:concentration}). Therefore, instead of $\ybar_0$, we shall use the empirical estimation $\ybar$ defined as follows
\begin{equation}\label{eq:defYbar}
	\ybar = \frac{1}{m}\sum_{i=1}^m y_i^2\a_i\a_i^T\indicator_{\{l\leq (y_i/\nu)\leq u\}}.
\end{equation}
Since truncation leads to bounded random variables, $\ybar$ has a higher concentration around its expectation than $\Y$. Thus, an eigenvector of $\ybar$ with a suitable length gives a better approximation to $\x$ than $\Y$. In this way, we obtain spectral initialization and its variants for general phase retrieval. It can be shown that $m=\Omega(n\log n)$ \cite[Lemma 5]{ma2018implicit} (resp. $m=\Omega(n)$ \cite[Proposition C.1]{chen2017solving}) samples are sufficient to guarantee the spectral initialization with $\Y$ (resp. $\ybar$) to obtain a desired initial guess, for general underlying signal $\x$ without sparsity assumption. Spectral methods based on $\Y$ and $\ybar$ are widely used and achieve order-wisely near-optimal sample complexity, but they may not provide the tightest estimates in terms of the over-sampling ratio $m/n$. An interesting question is how to perform so-called weak recovery by producing an estimator that is positively correlated with the underlying signal $\bm{x}$ using an optimal ratio $m/n$. One way to achieve this is by applying a preprocessing function $\Gamma(y_i)$ to $y_i$ and constructing the matrix $\sum_{i=1}^m\Gamma(y_i)\bm{a}_i\bm{a}_i^T$ instead of using $\Y$ or $\ybar$. By carefully examining the underlying statistical models, previous work \cite{luo2019optimal,mondelli2019fundamental} has provided an asymptotic characterization of spectral initialization, indicating that the preprocessing function $\Gamma(y_i)$ that optimizes the minimum ratio $m/n$ required for weak recovery can be determined analytically. In the following, we consider spectral methods under sparsity assumption to further reduce the sample complexity.

\subsection{Spectral initialization for sparse phase retrieval}\label{sec:specinit}
When $\x$ is $s$-sparse, the sample complexity can be reduced significantly by exploiting the sparsity prior of $\x$. The main idea is to separately estimate the support and non-zeros of $\x$. If the support $\calS=\supp{(\x)}$ is known or estimated, then $\EE([\Y]_{\calS})$ (or its variants such as $\EE([\ybar_0]_{\calS})$) has $\x_{\calS}$ as its principle eigenvectors. Based on this observation, spectral initialization in existing sparse phase retrieval methods first estimates an index set $\hat\calS$ as the support of the initial guess and then sets a rescaled principal eigenvector of $[\Y]_{\hat\calS}$ (or its variants) as non-zeros of the initial estimate (see \cite{wang2017sparse,jagatap2019sample}). Only a small principal submatrix of $\Y$ (or its variants) is involved in the procedure, and a high concentration of this small random submatrix suffices to give an accurate initialization. Since it is much easier to concentrate a small random matrix than a larger one, much fewer samples are required in spectral initialization for sparse phase retrieval.

In existing approaches \cite{wang2017sparse,jagatap2019sample}, the magnitudes of the diagonal entries of $\Y$ are used to estimate $\calS$. By a simple calculation, the expectation of diagonal entries of $\Y$ satisfy
\begin{align}\label{eq:EYjj}
	|[\EE\Y]_{jj}| = \left\{
	\begin{array}{rlc}
		&\ltwo{\x}^2+2x_j^2, &j \in \calS,\\
		&\ltwo{\x}^2, &j\in\calS^c.
	\end{array}
	\right.
\end{align}
Statistically, the diagonal entries on $\calS$ are larger than those on $\calS^c$. Thus, one simply chooses the set of indices of top-$s$ entries of $\Y$ as an estimation $\hat\calS$ of $\calS$, i.e.,
\begin{equation}\label{eq:hatSdiag}
	\hat\calS =  \{\text{indices of top $s$ elements of $[\Y]_{jj}$ in absolute value}\}.
\end{equation}

However, this approach of estimating $\hat\calS$ results in a spectral initialization with sample complexity $m=\Omega(s^2\log n)$ \cite{jagatap2019sample}. Though this sample complexity is better than $m=\Omega(n\log n)$ in general phase retrieval, it is still unnecessarily large and not optimal in $s$.

Actually, the sample complexity of this approach is governed by the gap between diagonals $[\EE\Y]_{jj}$ for $j\in\calS$ and $j\in\calS^c$ as in the following
\begin{equation}\label{eq:gamma0}
	\gamma_0: = \frac12\left(\min_{j\in\calS}|[\EE\Y]_{jj}| - \max_{j\in\calS^c}|[\EE\Y]_{jj}|\right) = \min_{j\in\calS}|x_j|^2.
\end{equation}
If $\max_j|[\EE\Y]_{jj}-[\Y]_{jj}|\leq \gamma_0$, then, for any $j\in\calS$ and $j'\in\calS^c$,
\begin{equation*}
	\begin{split}
		|[\Y]_{jj}|&\geq|[\EE\Y]_{jj}|-|[\EE\Y]_{jj}-[\Y]_{jj}|\geq
		|[\EE\Y]_{jj}|-\gamma_0\cr
		&\geq |[\EE\Y]_{j'j'}|+\gamma_0\geq|[\EE\Y]_{j'j'}|+|[\EE\Y]_{j'j'}-[\Y]_{j'j'}|
		\geq |[\Y]_{j'j'}|,
	\end{split}
\end{equation*}
which together with \eqref{eq:hatSdiag} implies $\hat\calS=\calS$; otherwise, those indices $j\in\calS$ satisfying $|[\EE\Y]_{jj}-[\Y]_{jj}|>\gamma_0$ might be missed in $\hat\calS$. Therefore, the gap $\gamma_0$ is the tolerance of the concentration error $|[\EE\Y]_{jj}-[\Y]_{jj}|$ for an accurate $\hat\calS$ and an accurate initial estimation. The larger gap $\gamma_0$, the larger tolerance of $|[\EE\Y]_{jj}-[\Y]_{jj}|$, and the fewer samples required. When random Gaussian measurement vectors are used, it can be shown that $\gamma_0$ in \eqref{eq:gamma0} finally leads to the sample complexity $m=\Omega(s^2\log n)$ for spectral initialization \cite{jagatap2019sample}.

\subsection{Truncated power method}\label{section:tpm}
To further reduce the sample complexity, we introduce a new initialization method, which is a truncated power method \cite{yuan2013truncated} with a modified spectral initialization.

\subsubsection{Sparse principal component analysis problem.}
Our algorithm is based on the following proposition, which characterizes $\x$ as the solution of a sparse eigenvector problem.
\begin{proposition}
Any solution to the following optimization problem is a multiple of $\x$:
\begin{equation}\label{eq:sparsePCA}
\max_{\w} \w^T(\EE\Y)\w\quad\text{s.t.}\quad\ltwo{\w} = 1,~\lzero{\w}\leq s,
\end{equation}
where $\Y$ is the matrix defined in \eqref{eq:defY}.
\end{proposition}
\begin{proof}
Let the support of the variable $\w$ in \eqref{eq:sparsePCA} be $\calT$. Then $\calT$ is a subset of $[n]$ satisfying $|\calT|\leq s$, and the objective function is rewritten as $\w^T(\EE\Y)\w=\w_{\calT}^T[\EE\Y]_{\calT}\w_{\calT}$ with $\ltwo{\w_{\calT}}=1$. Therefore, \eqref{eq:sparsePCA} is equivalent to
\begin{equation}\label{eq:sparsePCA2}
    \max_{|\calT|\leq s}~\max_{\ltwo{\bm{v}}=1}~\v^T[\EE\Y]_{\calT}\v.
\end{equation}
Using the variational property of eigenvalues of symmetric matrices, the optimal value of the inner maximization is the maximum eigenvalue of $[\EE\Y]_{\calT}$, which by \eqref{eq:defY} is $\ltwo{\x_{\calT}}^2+2$. Thus, \eqref{eq:sparsePCA} is further equivalent to
$$
\max_{|\calT|\leq s}~\ltwo{\x_{\calT}}^2+2.
$$
The optimal value of the above optimization is obviously $\ltwo{\x_{\calS}}^2+2$, which is attained at $\calT=\calS$. Since the maximum eigenvalue of $[\EE\Y]_{\calS}$ is simple and the corresponding eigenvectors are multiples of $\x_{\calS}$, solutions to \eqref{eq:sparsePCA2} are $\calT=\calS$ and $\bm{v}$ being multiples of $\x_{\calS}$, which eventually implies that solutions to the original problem \eqref{eq:sparsePCA} are multiples of $\x$.
\end{proof}

Like the spectral method, we replace the expectation in \eqref{eq:sparsePCA} with its truncated empirical version $\ybar$, and we solve
the following quadratic maximization problem on the unit sphere with a sparsity constraint:
\begin{align}\label{prob:sparseeigenvalue}
	\max~\w^T\ybar\w,\quad\text{s.t.}\quad \ltwo{\w} = 1,~\lzero{\w}\leq s.
\end{align}
Problem \eqref{prob:sparseeigenvalue} aims to find an $s$-sparse principal eigenvector of $\ybar$. It is called a sparse principal component analysis (sparse PCA) problem and has been studied in the literature \cite{moghaddam2005spectral,journee2010generalized,ma2013sparse,yuan2013truncated}. Let $\x_0$ be $\nu$ (defined in \eqref{eq:defnu}) times a solution to the problem~\eqref{prob:sparseeigenvalue}. Then, it has been shown in~\cite{liu2021towards} that $\Omega(s\log n)$ samples are sufficient to guarantee $\x_0$ is a desired initial guess for non-convex algorithms.
Various algorithms are available for solving \eqref{prob:sparseeigenvalue} in the studies, e.g., \cite{moghaddam2005spectral,journee2010generalized,ma2013sparse,yuan2013truncated,balakrishnan2017computationally}. However, the sparse PCA problem is known to suffer from a statistical-to-computational gap, and any computationally efficient method must pay a statistical price, even under natural distributional assumptions \cite{berthet2013complexity}. For instance, in \cite{balakrishnan2017computationally}, a convex relaxation is used to overcome the computational challenge and achieve a practical (polynomial time) algorithm, but this approach requires $\Omega(s^2\log(n/s))$ samples. Therefore, the proposed algorithm for sparse PCA also faces a statistical-computational trade-off and has a bottleneck of $\Omega(s\bar{s}_x \log n)$.  Meanwhile, the algorithms for sparse PCA are not tailored for our problem, and their empirical and theoretical performance is not clear in sparse phase retrieval.

\subsubsection{Truncated power iteration.}
We adopt the truncated power method \cite{yuan2013truncated} to solve the sparse PCA problem \eqref{prob:sparseeigenvalue}. We also provide the theoretical sample complexity to produce the desired initialization for non-convex sparse phase retrieval algorithms, showing that the truncated power method achieves a nearly optimal sample complexity when the underlying signal contains only very few significant components.

The truncated power method extends the popular power method to the case where the target principle eigenvector is sparse. The standard power method is to find a principal eigenvector of a given matrix, say $\ybar$, by solving \eqref{prob:sparseeigenvalue} without the sparsity constraint. The iteration is
\begin{equation}\label{eq:powermethod}
\w_t = \ybar\w_{t-1}^0,\quad \w_{t}^0=\frac{\w_t}{\ltwo{\w_{t}}}.
\end{equation}
It can be proved that, under mild assumptions, the angle between $\w_t$ and the principal eigenspace of $\ybar$ converges to $0$. However, the power method generally gives a dense principal eigenvector. To obtain a sparse principal eigenvector via \eqref{prob:sparseeigenvalue}, we modify \eqref{eq:powermethod} by applying a truncation operator to sparsify the iteration vector. More specifically, we define a truncation operator
\begin{align*}%\label{def:Ts}
	T_s(\w):~ \text{keep the top $s$ entries of $\w$ in absolute values and set others to be}~0,
\end{align*}
and generate a sequence $\{\w_t^0\}_{t=0,1,\ldots.}$ by
\begin{equation}\label{eq:itertp}
\w_t = T_s\left(\ybar\w_{t-1}^0\right),\quad \w_{t}^0=\frac{\w_t}{\ltwo{\w_{t}}}.
\end{equation}
The above iteration is the truncated power method \cite{yuan2013truncated}. At each iteration, the truncated power method uses the truncation operator $T_s$ to ensure that the iteration vector is $s$-sparse while aligning with the principal eigenspace of an order-$s$ principal minor of $\ybar$.

\subsubsection{Initial \texorpdfstring{$\w_0^0$}.}\label{sec:initfortpm}
Provided a proper $\w_0^0$, linear convergence of the truncated power method has been proved in \cite{yuan2013truncated} under certain conditions. However, the theoretical result developed in \cite{yuan2013truncated} is for general purpose, and the conditions for convergence cannot work trivially for our problem. Also, the requirement on $\w_0^0$ is stringent and does not link directly to the sample complexity. Moreover, the efficiency of the truncated power method depends on $\w_0^0$, and it is challenging to construct a good $\w_0^0$.

To address this, we introduce a new method to produce the initial $\w_0^0$ for the truncated power method in our case. We aim to find a unit vector $\w_0^0$ to align with $\x$ well. Spectral initialization \cite{wang2017sparse,jagatap2019sample} presented in \cref{sec:specinit} is also applicable here but leads to unnecessarily large sample complexity, as mentioned. Therefore, we modify the spectral initialization to produce $\w_0^0$. Recall that the spectral initialization first estimates the support of $\x$ by taking the indices set $\hat{\calS}$ of the top $s$ diagonal entries of $\Y$ (see \eqref{eq:EYjj} and \eqref{eq:hatSdiag}) and then use a rescaled eigenvector of $[\ybar]_{\hat{\calS}}$ to approximate non-zero entries of $\x$. Following \cite{wu2021hadamard}, instead of the diagonal entries, we use the indices of the $s$ largest number of entries of $\Y\e_{j_0}$ in absolute value as $\hat\calS$, where $j_0$ is the index that $Y_{jj}$ achieves maximum. Our choice of $\hat{\calS}$ is based on the observation that
\begin{align}\label{eq:EYj0}
	|[\EE\Y\e_{j_0}]_j| = \left\{
	\begin{array}{rlc}
		&2x_{j_0}^2+\ltwo{\x}^{2}, &j = j_0,\\
		&2|x_{j_0}x_j|, &j\in\calS\backslash\{j_0\},\\
		&0, &j\in\calS^c.
	\end{array}
	\right.
\end{align}
The result in \eqref{eq:EYj0} indicates that $|[\Y\e_{j_0}]_j|$ is statistically larger at indices $j\in \calS$ than that at $j\in \calS^c$, as long as $j_0\in\calS$ (which is proved later in \cref{lemma:xj0} under mild conditions).  Thus, we use the estimation
\begin{equation}\label{eq:hatS}
	\hat\calS =  \{\text{indices of top $s$ elements of $\Y\e_{j_0}$ in absolute value}\}
\end{equation}
as the support of $\w_0^0$. Then, similar to the standard spectral initialization, we set a unit principle eigenvector of $[\ybar]_{\hat\calS}$ as non-zeros of $\w_0^0$.

Compared to the spectral initialization, \eqref{eq:hatS} gives a more accurate $\hat\calS$ and uses fewer samples. Similar to the argument in \cref{sec:specinit}, the sample complexity of our method for $\w_0^0$ depends on the following gap
$$
\gamma_1 = \frac12\left(\min_{j\in\calS}|[\EE\Y\e_{j_0}]_j| - \max_{j\in\calS^c}|[\EE\Y\e_{j_0}]_j|\right) = |x_{j_0}|\cdot\min_{j\in\calS}|x_j|.
$$
A larger gap $\gamma_1$ leads to a smaller sample complexity. As we will show, we can always find $j_0$ such that $|x_{j_0}|\geq \frac{1}{2}\linf{\x}$ with high probability. When $\linf{\x}$ is significantly larger than $\min_{j\in\calS}|x_j|$, the gap $\gamma_1$ depends only linearly on $\min_{j\in\calS}|x_j|$, which improves order-wisely the gap $\gamma_0$ in the original spectral initialization. Therefore, our method for $\w_0^0$ requires order-wisely fewer samples than the original spectral initialization. Indeed, we will show later that, when $\linf{\x}$ is as large as $\calO(\ltwo{\x})$, our method requires only $\Omega(s\log n)$ samples to produce a good $\w_0^0$ for the truncated power method and even for non-convex sparse phase retrieval algorithms.

\subsubsection{Full algorithm of the truncated power method.}
With the initial $\w_0^0$, we then perform the main iteration \eqref{eq:itertp} in the truncated power method. To improve the accuracy, we slightly increase the sparsity from $s$ to $s'$ during the main iteration to search the principal component on slightly larger support. More explicitly, we replace $T_s$ in \eqref{eq:itertp} with $T_{s'}$. After the main iteration, we need to project the result back to the $s$-sparse set to obtain the final output of our truncated power iteration as the initialization of any non-convex sparse phase retrieval algorithm. The details of the full algorithm are summarized in \cref{alg:init}.
\begin{algorithm}[H]
	\caption{Truncated Power Method Initialization}\label{alg:init}
	\begin{algorithmic}
		\STATE{\textbf{Input } Measurement vectors and observations $\{\a_i,y_i\}_{i=1}^m$, the sparsity $s$, parameter $s'$, maximum iteration $t_{\max}$.}
		%		\STATE{$\backslash\backslash$ Initialization}
		\STATE{Compute $\{Y_{jj}\}_{j=1}^n$ and set $j_0 = \arg\max_{1\leq j\leq n} Y_{jj}$.}
		\STATE{Set $\hat\calS$ be the set with top $s$ elements of $\Y\e_{j_0}$ in absolute value.}
		\STATE{Form $[\ybar]_{\hat\calS}$ and set $\w_0^0$ be the top singular vector of the matrix $[\ybar]_{\hat\calS}$ of unit length.}
		%		\STATE{$\backslash\backslash$ Truncated power method}
		\FOR{$t = 1,\ldots,t_{\max}$}
		% \STATE{$\w_t' = \ybar\w_{t-1}/\ltwo{\ybar\w_{t-1}}$}
		\STATE{$\w_t = T_{s'}(\ybar\w_{t-1}^0)$}
		\STATE{$\w_t^0 = \w_t/\ltwo{\w_t}$}
		\ENDFOR
		%		\STATE{$\backslash\backslash$ Projection onto $\MM_s$}
		\STATE{$\hat\x^0 = T_s(\w_{t_{\max}}^0)$}
		%		\STATE{\textbf{Output } $\x_0 = T_s(\w_{t_{\max}})$}
		\STATE{\textbf{Output } $\hat\x = \nu\hat\x^0$ with $\nu$ in \eqref{eq:defnu}. }
	\end{algorithmic}
\end{algorithm}

\subsection{Modified spectral initialization}\label{section:ossi}
As mentioned in \cref{sec:initfortpm}, the modified spectral initialization for the truncated power method improves the standard spectral initialization for non-convex sparse phase retrieval. Therefore, the modified spectral initialization can also be used to initialize non-convex sparse phase retrieval. We obtain a modified spectral method summarized in \cref{alg:init:onestep}. A similar idea by using the set with top $s$ elements of $\Y\e_{j_0}$ in absolute value was adopted in \cite{wu2021hadamard}. However, their theory was established based on some lower bound assumption on $x_{\min}$, which allows only a very small dynamic range of the nonzero entries of the underlying signal. We will show that this assumption is not necessary.
\begin{algorithm}[H]
	\caption{Modified Spectral Initialization}\label{alg:init:onestep}
	\begin{algorithmic}
		\STATE{\textbf{Input } Measurement vectors and observations $\{\a_i,y_i\}_{i=1}^m$, the sparsity $s$.}
		%		\STATE{$\backslash\backslash$ Initialization}
		\STATE{Compute $\{Y_{jj}\}_{j=1}^n$ and set $j_0 = \arg\max_{1\leq j\leq n} Y_{jj}$.}
		\STATE{Set $\hat\calS$ be the set with top $s$ elements of $\Y\e_{j_0}$ in absolute value.}
		\STATE{Form $[\ybar]_{\hat\calS}$ and set $\hat{\x}^0$ be the top singular vector of the matrix $[\ybar]_{\hat\calS}$ of unit length.}
		\STATE{\textbf{Output } $\hat\x = \nu\hat\x^0$ with $\nu$ in \eqref{eq:defnu}.}
	\end{algorithmic}
\end{algorithm}
As we shall see in the rest of the paper, theoretical guarantees and numerical experiments confirm that \cref{alg:init:onestep} can also achieve a nearly optimal sample complexity but with a larger constant than \cref{alg:init}, under a suitable assumption on the underlying signal.

\section{Theoretical Results}\label{section:theory}
In this section, we provide theoretical results on both Algorithms \ref{alg:init} and \ref{alg:init:onestep}. We will show that our algorithms are guaranteed to provide the desired estimation for non-convex sparse phase retrieval algorithms with a low sample complexity. We also discuss the significance of our theoretical results compared to existing results in the literature.

\subsection{Theoretical guarantee and sample complexity}
Like most of the theories of non-convex sparse phase retrieval, we assume that the measurement vectors $\{\a_i\}_{i=1}^{m}\subset\mathbb{R}^n$ are i.i.d. random Gaussian with mean the $\bm{0}$ vector and covariance the identity matrix.
%Recall we assume $\x$ is normalized such that $\ltwo{\x} = 1$.
Under this assumption, the refinement stage of most non-convex sparse phase retrieval algorithms requires an initialization $\hat\x$ that is $\delta$-close to $\x$, i.e., $\dist(\x,\hat\x)\leq \delta\ltwo{\x}$ for some universal constant $\delta\in(0,1)$.

To present our result, we recall that the \emph{stable sparsity} of $\x$ is defined by $\stablesparsity{\x}:= \frac{\ltwo{\x}^2}{\linf{\x}^2}$,
and it is clear that $\stablesparsity{\x}\in[1,s]$. When there are only very few non-zero entries significantly larger than others in absolute value, the stable sparsity counts only large non-zero entries but not small ones. Therefore, stable sparsity is a more accurate quantity than sparsity to describe the number of significant non-zero entries of a vector.

Since the modified spectral initialization \cref{alg:init:onestep} is a simplification of the truncated power method initialization \cref{alg:init}, we first give our theoretical result on the simpler version \cref{alg:init:onestep}.  Our result states that the output of Algorithm \ref{alg:init:onestep} with $m=\Omega\left(\blue{\max\{\delta^{-1}(\stablesparsity{\x}s)^{1/2}\log^3n, \delta^{-2}\stablesparsity{\x}s\log n\}}\right)$ Gaussian measurements is guaranteed to stay in a $\delta$-neighborhood of $\x$ with high probability.  We summarize the result in the following theorem, whose proof is relegated to \cref{proof:thm:onestep}.
\begin{theorem}[Desired initialization guarantee of \cref{alg:init:onestep}]\label{thm:onestep}
	Let $\x\in\RR^n$ be an $s$-sparse vector. Let $y_i = |\inp{\a_i}{\x}|$, $i=1,\ldots,m$, be $m$ phaseless measurements of $\x$ without noise, where $\a_i\sim N(0,\I_n)$, $1\leq i\leq m$. There exist universal constants $l$, $u$, and $C>0$ such that: for any $\delta\in(0,1)$, if
%	$m \geq C\max\{\delta^{-2}\stablesparsity{\x}s,\delta^{-4}\}\log n$,
	\blue{$m\geq C\max\{\delta^{-1}(\stablesparsity{\x}s)^{1/2}\log^3n, \delta^{-2}\stablesparsity{\x}s\log n\},$}
	then with probability exceeding $1-5m^{-1}n^{-1}$, the output of Algorithm \ref{alg:init:onestep} with truncation parameters $l,u$ satisfies $\dist(\x,\hat{\x}) \leq \delta\ltwo{\x}$.
\end{theorem}

We then give our theoretical result on \cref{alg:init}, the truncated power method initialization.
We summarize the result in the following \cref{thm:init} and postpone the proof to \cref{sec:proof:tp}. Our result reveals that
we can further reduce the sample complexity to $m=\Omega\left(\blue{\max\{(\stablesparsity{\x}s)^{1/2}\log^3n, \stablesparsity{\x}s\log n,\delta^{-2}s\log n\}}\right)$ to produce a $\delta$-close initial estimation by using \cref{alg:init}.
\begin{theorem}[Desired initialization guarantee of \cref{alg:init}]\label{thm:init}	
	Let $\x\in\RR^n$ be a $s$-sparse vector. Let $y_i = |\inp{\a_i}{\x}|$, $i=1,\ldots,m$, be $m$ phaseless measurements of $\x$ without noise, where $\a_i\sim N(0,\I_n)$, $1\leq i\leq m$. There exist universal constants $l$, $u$, and $C,C_1>0$ such that: for any $\delta\in(0,1)$, if
%	$m \geq C\cdot\max\{\delta^{-4},\delta^{-2}s,\stablesparsity{\x}s\}\log n $,
	\blue{$m\geq C\max\{(\stablesparsity{\x}s)^{1/2}\log^3n, \stablesparsity{\x}s\log n,\delta^{-2}s\log n\}$}
	then with probability exceeding $1-5m^{-1}n^{-1}$, the output of Algorithm \ref{alg:init:onestep} with parameters $l,u$, $s'=C_1s$, and $t_{\max}= \calO(\log(\delta^{-1}))$ satisfies $\dist(\x,\hat{\x}) \leq \delta\ltwo{\x}$.
\end{theorem}

\blue{
If we focus on the regime when $s= \Theta(n^{\sigma})$ for some $\sigma\in(0,1)$, then the required sample size in \cref{alg:init:onestep} and \cref{alg:init} are respectively $\Omega(\delta^{-2}\stablesparsity{\x}s\log n)$ and $\Omega(\max\{\stablesparsity{\x},\delta^{-2}\}s\log n)$.
}
In most non-convex sparse phase retrieval algorithms, the radius $\delta$ of the local convergence basin in the refinement stage is a small constant (e.g., for HTP~\cite{cai2022sparse}, $\delta=\frac{1}{8}$).
\blue{
When $\stablesparsity{\x}\geq \delta^{-2}$, the sample complexity is improved from $\Omega(\delta^{-2}\stablesparsity{\x}s\log n)$ to $\Omega\left(\stablesparsity{\x}s\log n\right)$.
}
In this case, the improvement by \cref{alg:init} is significant --- the coefficient before $s\log n$ is reduced from $\delta^{-2}\stablesparsity{\x}$ to \blue{$\stablesparsity{\x}$}.
\blue{
In general, the coefficient in front of $s\log n$ is enhanced from $\stablesparsity{\x}\delta^{-2}$ to the maximum of $\stablesparsity{\x}$ and $\delta^{-2}$.
}
This improvement makes \cref{alg:init} has better numerical performance than \cref{alg:init:onestep}. Therefore, in the following, we mainly compare \cref{alg:init} with others.

As mentioned earlier above, the global sample complexity of most two-stage non-convex algorithms is dominated by the initialization stage. Therefore,  the sample efficiency of many two-stage algorithms can be improved when combined with the proposed initialization methods. Examples of such algorithms are SPARTA~\cite{wang2017sparse}, CoPRAM~\cite{jagatap2019sample}, thresholding/projected Wirtinger flow~\cite{cai2016optimal,soltanolkotabi2019structured}, SAM~\cite{cai2022sample} and HTP~\cite{cai2022sparse}, to just name a few. We use the two-stage hard thresholding pursuit (HTP) algorithm \cite{cai2022sparse}  as a typical example to illustrate this. When combined with the local refinement algorithm HTP \cite[Algorithm 1]{cai2022sparse}, we then have the two-stage algorithm (described in \cref{alg:htp} for completeness).  Also, to empirically get a better estimate, similar to \cite{wu2021hadamard}, for the TP and Modified Spectral method, we can further implement a multiple-restarted version. The details are given in Algorithm \ref{alg:tpmr}. We give the sample complexity of the algorithm in the following corollary.

\begin{corollary}\label{coro:HTP}
Assume the measurement vectors $\{\a_i\}_{i=1}^m$ are i.i.d. Gaussian random and the phaseless measurements $\{y_i\:=|\langle\a_i,\x\rangle|\}_{i=1}^{m}$ are noiseless. Then, \cref{alg:htp} is guaranteed to have an exact recovery of $\x$ in at most $K\le \calO(\log m + \log(\|\bm{x}\|_2/|x_{\min}|))$ local refinement iterations with high probability, provided \blue{$m=\Omega(\max\{(\stablesparsity{\x}s)^{1/2}\log^3n, \stablesparsity{\x}s\log n,\delta^{-2}s\log n\})$}.
\end{corollary}
\begin{proof}
It is a direct consequence of \cite[Theorem 1]{cai2022sparse} and \cref{thm:init}.
\end{proof}

Compared to the recovery guarantee in \cite[Theorem 2]{cai2022sparse}, which requires $m=\Omega(\delta^{-2}s^2\log(n))$, \cref{coro:HTP} indicates that our approach is more sample-efficient. In particular, \blue{when $s= \Theta(n^{\sigma})$ for some $\sigma\in(0,1)$} and $\x$ contains only a few large entries in absolute value so that $\stablesparsity{\x}=\calO(1)$, the sample complexity of our algorithms is $\Omega(s\log n)$, which is nearly optimal. Since our methods \cref{alg:init} and \cref{alg:init:onestep} can initialize the refinement stage of any aforementioned non-convex sparse phase retrieval algorithms, the improvement in sample complexity also holds for them.

\begin{algorithm}[hbt!]
   \caption{Two-Stage HTP with TP Initialization}\label{alg:htp}
   \begin{algorithmic}[1]
     \STATE Input: Data $\left\{\bm{a}_i,y_i\right\}_{i=1}^{m}$, step size $\mu>0$ (e.g., $\mu=0.95$).
     \STATE \textcolor{OliveGreen}{// Initialization:}
     \STATE Let $\bm{x}_{0}=\hat\x$ be the initial guess produced by  \cref{alg:init}.

     \STATE \textcolor{OliveGreen}{// Local refinement updates:}
     \FOR{ $k=0,1,\ldots,K-1$ }
     \STATE $\bm{z}_{k+1}=\bm{A}\bm{x}_k$, \blue{where $\bm{A} = [\a_1,\ldots, \a_m]^T$}
     %\STATE $\bm{y}_{k+1}=\bm{y} \odot \mathrm{sgn}{\left(\bm{z}_{k+1}\right)}$
     \STATE $\calS_{k+1}=\mathrm{supp}\Big(T_s\left(\bm{x}_{k}+\mu \bm{A}^T\left( \bm{y} \odot \mathrm{sgn}{\left(\bm{z}_{k+1}\right)}-\bm{z}_{k+1}\right)\right)\Big)$
     \STATE $\bm{x}_{k+1}=\mathop{\arg\, \min}\limits_{\mathrm{supp}\left(\bm{x}\right)\subseteq \calS_{k+1}}\frac{1}{2}\lVert \bm{A}\bm{x}-\bm{y} \odot \mathrm{sgn}{\left(\bm{z}_{k+1}\right)}\lVert_{2}^{2}$
     \ENDFOR
     \STATE Output $\bm{x}_K$.
    \end{algorithmic}
\end{algorithm}
\begin{algorithm}
	\caption{Truncated Power Method with Multiple Restarts (TP-MR)}\label{alg:tpmr}
	\begin{algorithmic}
		\STATE{\textbf{Input } Matrix $\Y,\ybar$,  $\bm{A} = [\a_1,\ldots, \a_m]^T$, sparsity $s$, maximum steps for truncated power method $t_{\max}$, restarted number $b$.}
		\STATE{Compute $\{Y_{jj}\}_{j=1}^n$.}
		\FOR{$b' = 1,\ldots, b$}
		\STATE{Set $j_0$ be the $b'$-th largest entry in $\{Y_{jj}\}_{j=1}^n$.}
		\STATE{Run TP (\cref{alg:init}) with Modified Spectral initialization to get $\hat\x_{b'}$.}
		\STATE{Run HTP with initialization $\hat\x_{b'}$ to get $\x_{b'}$.}
		\ENDFOR		
		\STATE{Set $b_{\min}$ to be the index minimizing $\ltwo{\A^T(\A\x_{b'} - |\y|\cdot\sign(\A\x_{b'}))}$.}
		\STATE{\textbf{Output } $\x_{b_{\min}}$}

	\end{algorithmic}
\end{algorithm}
\subsection{Discussion on sample complexity}

The sample complexity of \cref{alg:init} and \cref{alg:init:onestep} beats all existing initialization algorithms. We show in \cref{tab:samplecomlexity} a summary of the comparison of different initialization algorithms \blue{under the regime when $s= \Theta(n^{\sigma})$ for some $\sigma\in(0,1)$}. Since \cref{alg:init} has a lower sample complexity than \cref{alg:init:onestep}, we compare only \cref{alg:init} with other existing initialization methods.

To make the comparison more clearly, we assume $\blue{\stablesparsity{\x}}\geq \delta^{-2}$ so that the sample complexity of our truncated power method initialization \cref{alg:init} is $m=\blue{\Omega(\stablesparsity{\x}s\log n)}$ to produce a $\delta$-close approximation to $\x$. In \cite{wang2017sparse}, the sample complexity of the spectral initialization is $m=\Omega(\delta^{-2}s^2\log(n))$ with an additional assumption $\min_{i\in\calS}|x_i|\geq \ltwo{\x}/\sqrt{s}$ on the underlying signal. Later, \cite{jagatap2019sample} removed the additional assumption on the minimum entry. Since $\stablesparsity{\x}\leq s$, our \cref{alg:init} requires fewer samples than the spectral initialization. Especially, when there are only very few significant components in $\x$ such that $\stablesparsity{\x}=\calO(1)$, the sample complexity of \cref{alg:init} is $m=\Omega(s\log n)$, which is optimal in $s$ and order-wisely better than $m=\Omega(s^2\log(n))$ in \cite{wang2017sparse,jagatap2019sample}. Recently, \cite{wu2021hadamard} proposed a new initialization method that is a multiple run of \cref{alg:init:onestep} combined with a few steps of Hardmard Wirtinger flow, and the sample complexity of its initialization is $m=\Omega(\delta^{-2}\stablesparsity{\x} s\log(n))$ under the assumption that $\min_{i\in\calS}|x_i|\geq \ltwo{\x}/\sqrt{s}$. On the contrary, \cref{alg:init} not only achieves the same order of sample complexity without any assumption on $\x$, but also improves dependency on $\delta$ and $\stablesparsity{\x}$ from $\delta^{-2}\stablesparsity{\x}$ to $\max\{\delta^{-2},\stablesparsity{\x}\}$. Since $\delta$ is usually a small constant in various non-convex algorithms, this improvement is significant, which makes \cref{alg:init} initialized algorithms outperform others.

The sample complexity of our algorithms is $\Omega(\stablesparsity{\x}s\log n)$, which contains an additional term $\stablesparsity{\x}$ compared to the unconditionally optimal sample complexity $\Omega(s\log n)$. It is unlikely that we can remove this additional term. Indeed, similar phenomenon exists ubiquitously non-convex algorithms in the literature, such as low rank matrix recovery \cite{wei2016guarantees,tong2021accelerating} or low rank tensor recovery \cite{cai2022generalized,shen2022computationally,tong2022scaling}. For example, in the matrix case, we can analogously define the stable rank of a rank-$r$ matrix $\M\in\RR^{n\times n}$ as $\stablerank(\M) := \fro{\M}^2/\op{\M}^2$. It is trivial to see $\stablerank(\M)\in[1,r]$. Although the degree of freedom of a rank $r$ matrix is of order $O(nr)$, the required sample size for all existing non-convex low-rank matrix recovery algorithms is of order $O(nr^k)$ for some $k\geq 2$. The reason is that people usually bound $\stablerank(\M)$ simply by $r$, resulting in a sub-optimal order on $r$. In the low-rank matrix recovery case, we usually focus on the regime when $r$ is much smaller than $d$, and the sample size is optimal in $d$ no matter $k$. While in the sparse phase retrieval case, the trivial bound on $\stablesparsity{\x}$ will result in a sub-optimal order on $s$. Considering the analogy between the stable rank and stable sparsity, we conjecture that there is an information-theoretical gap between the degree of freedom of $\x$ and the actual required sample size. Despite the pessimistic existence of the gap, our result here requires a much fewer sample size compared to \cite{phase15,cai2016optimal,wang2017sparse,jagatap2019sample} and removes the restriction on $\x$ compared with \cite{wu2021hadamard}.

\section{Proofs}\label{section:proofs}
In this section, we present the proofs of Theorem \ref{thm:onestep} and Theorem \ref{thm:init}. We first introduce some technical lemmas in Section \ref{section:lemmas}. Then Theorem \ref{thm:onestep} and Theorem \ref{thm:init} are proved in \cref{proof:thm:onestep} and \cref{sec:proof:tp} respectively.

We introduce some notations that will be used throughout. Recall that the matrices $\Y$, $\ybar$, and $\ybar_0$ are defined in \eqref{eq:defY}, \eqref{eq:defYbar}, and \eqref{eq:defYbar0} respectively.
Then \blue{from Lemma \ref{lemma:expectation}}, we have $\EE\ybar_0 = (\beta - \alpha)\x\x^T + \alpha\ltwo{\x}^2\I$, where $\beta = \EE g^4\indicator_{\{l\leq |g|\leq u\}}$, $\alpha = \EE g^2\indicator_{\{l\leq |g|\leq u\}}$, and $g\sim N(0,1)$. It can be easily verified there exist universal constants $u\geq 1$ and $l\geq 0$ such that $\beta/\alpha \geq 2$, $\alpha\geq 1/2$ \blue{(for example, if we set $u=10,l=1/2$, then $\beta\approx 2.995,\alpha\approx 0.969$)}. We shall proceed with the proofs under this choice of $u,l$. We define a matrix
$$
\E = \ybar - \EE\ybar_0.
$$
We denote $\x^0 = \x/\ltwo{\x}$, the normalized version of $\x$.

\subsection{Technical Lemmas}\label{section:lemmas}
This section presents some technical lemmas that will be used to prove the theorems.
\blue{
\begin{lemma}\label{lemma:triangleinequality:dist}
	Let $\u_i\in\RR^n$ for $i = 1,2,3$, then
	$$\dist(\u_1,\u_2)\leq \dist(\u_1,\u_3) + \dist(\u_2,\u_3).$$
\end{lemma}
\begin{proof}
	Denote $\sigma_{1,3}:=\arg\min_{\sigma\in\{\pm 1\}}\ltwo{\u_1 - \u_3}$. Then using triangle inequality,
	\begin{align*}
		\dist(\u_1,\u_2) &= \min_{\sigma\in\{\pm 1\}}\ltwo{\u_1-\sigma \u_2}\leq \min_{\sigma\in\{\pm 1\}}\big[\ltwo{\u_1-\sigma_{1,3} \u_3} + \ltwo{\sigma_{1,3}\u_3 - \sigma \u_2}\big]\\
		&= \dist(\u_1,\u_3) + \min_{\sigma\in\{\pm 1\}}\ltwo{\u_2-\sigma_{1,3}\sigma\u_3} =  \dist(\u_1,u_3) + \dist(\u_2,\u_3).
	\end{align*}
And this finishes the proof.
\end{proof}

}

Next we then give the concentration bound for $\nu$ defined in \eqref{eq:defnu}.
\begin{lemma}\label{lemma:nu:concentration}
    With probability exceeding $1-\frac{2}{mn}$, $\nu: = \left(\frac{1}{m}\sum_{i=1}^my_i^2\right)^{1/2}$ has the following concentration:
    $$
|\nu^2 - \ltwo{\x}^2| \leq 3\sqrt{\frac{\log(mn)}{m}}\ltwo{\x}^2.
$$
\end{lemma}
\begin{proof}
    Using the concentration for sum of $\chi^2$ random variables (see Lemma 4.1 in \cite{laurent2000adaptive}), we see
    \begin{align*}
        \PP\bigg(|\nu^2 - \ltwo{\x}^2|\leq 3\max\left\{\sqrt{t/m},t/m\right\}\ltwo{\x}^2\bigg)\geq 1-2e^{-t}.
    \end{align*}
    Now setting $t=  \log(mn)$ and we get the desired result.
\end{proof}

\blue{
\begin{lemma}\label{lemma:expectation}
	Let $y = |\inp{\a}{\x}|$ for $\a\sim N(0,\I_n)$. For any $\alpha_2\geq \alpha_1\geq0$, we have
	$$\EE y^2\a\a^T\cdot\indicator_{\{\alpha_1\ltwo{\x}\leq y\leq \alpha_2\ltwo{\x}\}} = (\gamma_4 -\gamma_2)\x\x^T + \ltwo{\x}^2\gamma_2\I_n,$$
	where $\gamma_k =\EE g^k\indicator_{\{\alpha_1\leq |g|\leq \alpha_2\}}$ for $k\in\{2,4\}$ and $g\sim N(0,1)$.
\end{lemma}
\begin{proof}
	Without loss of generality we may assume $\ltwo{\x}=  1$. Let $g = \inp{\a}{\x}$, then $g\sim N(0,1)$. And the covariance between $a_i$ and $g$ is
	$$\text{Cov}(a_i,g) = \EE a_ig - \EE a_i\EE g = \EE a_i\sum_{j=1}^na_jx_j = x_i.$$
	Therefore we may write $a_i$ as $a_i = x_ig + t_i$ for some $t_i\sim N(0,1-x_i^2)$ that is independent of $g$. Now we first compute the expectations of diagonal entries. For any $i\in[n]$,
	\begin{align*}
		\EE y^2a_i^2 \indicator_{\{\alpha_1\ltwo{\x}\leq y\leq \alpha_2\ltwo{\x}\}}  &= \EE g^2(x_ig + t_i)^2 \cdot\indicator_{\{\alpha_1\leq |g|\leq \alpha_2\}}\\
		&= (x_i^2\EE g^4 +\EE t_i^2\EE g^2) \cdot\indicator_{\{\alpha_1\leq |g|\leq \alpha_2\}}\\
		&= x_i^2\gamma_4 + (1-x_i^2)\gamma_2.
	\end{align*}
As for the off-diagonal entries, for any $i\neq j \in [n]$, we have
\begin{align*}
	\EE y^2a_ia_j \indicator_{\{\alpha_1\ltwo{\x}\leq y\leq \alpha_2\ltwo{\x}\}}  &= \EE g^2(x_ig + t_i)(x_jg + t_j) \cdot\indicator_{\{\alpha_1\leq |g|\leq \alpha_2\}}\\
	%&= (x_ix_j\EE g^4 +\EE t_it_j\EE g^2) \cdot\indicator_{\{\alpha_1\leq |g|\leq \alpha_2\}}\\
	&= x_ix_j\gamma_4 -x_ix_j\gamma_2,
\end{align*}
where in the last inequality we have used the fact that
$
	\EE t_it_j  = \EE(a_i-x_ig)(a_j-x_jg) = -x_ix_j.
$
And we conclude $\EE y^2\a\a^T\indicator_{\{\alpha_1\leq y\leq \alpha_2\}} = (\gamma_4 -\gamma_2)\x\x^T +\gamma_2\I_n.$

\end{proof}}

Next, we would like to bound the quantity $\rho(\E,s)$, which is defined in \eqref{eq:defrho} and is the maximum spectral norm of all order-$s$ principal minors of $\E$.
\begin{lemma}\label{lemma:ssparseeigenvalue}
	With probability exceeding $1- n^{-s}-\frac{2}{mn}$, for any $t>0$, there exists absolute constant $C>0$ such that
	$$\rho(\E,s) \leq 7t\ltwo{\x}^2$$ holds as long as
%	$m\geq C_2\cdot\max\left\{s\frac{u^4}{t^2},s\frac{u^2}{t},\frac{u^2}{t^4}\right\}\cdot\log n$
	$m\geq Cs\log n\max\{\frac{u^4}{t^2},\frac{u^2}{t}\}$.
\end{lemma}

\blue{
\begin{proof}
	We decompose $\E$ as
	$$\E = \underbrace{\bar\Y -\bar \Y_0}_{\E_1} + \underbrace{\bar\Y_0 -\EE\bar\Y_0}_{\E_2}.$$
	We first consider $\rho(\E_1,s)$. Recall
	$$\E_1 = \frac{1}{m}\sum_{i=1}^my_i^2\a_i\a_i^T\bigg(\indicator_{\{l\nu\leq y_i\leq u\nu\}} - \indicator_{\{l\ltwo{\x}\leq y_i\leq u\ltwo{\x}\}}\bigg).$$
	Consider the event
	$$\calE = \bigg\{(1-c_0)\ltwo{\x}\leq \nu\leq (1+c_0)\ltwo{\x}, \text{~where~} c_0 = 3\bigg(\frac{\log (mn)}{m}\bigg)^{1/2}\bigg\}.$$
	From Lemma \ref{lemma:nu:concentration}, we see $\calE$ holds with probability exceeding $1-\frac{2}{mn}$.
	Then under $\calE$, we have
	\begin{equation*}
		\begin{aligned}
			&\quad \indicator_{\{l\nu\leq y_i\leq u\nu\}} - \indicator_{\{l\ltwo{\x}\leq y_i\leq u\ltwo{\x}\}}\\
			& \leq \indicator_{\{l(1-c_0)\ltwo{\x}\leq y_i\leq u(1+c_0)\ltwo{\x}\}} - \indicator_{\{l\ltwo{\x}\leq y_i\leq u\ltwo{\x}\}}\\
			&= \left\{\begin{array}{rlc}
				&1, &\text{~if~}y_i\in[l(1-c_0)\ltwo{\x}, l\ltwo{\x}]\cup[u\ltwo{\x}, u(1+c_0)\ltwo{\x}],\\
				&0, &\text{~otherwise,}
			\end{array}
			\right.
		\end{aligned}
	\end{equation*}
and
\begin{equation*}
	\begin{aligned}
		&\quad \indicator_{\{l\nu\leq y_i\leq u\nu\}} - \indicator_{\{l\ltwo{\x}\leq y_i\leq u\ltwo{\x}\}}\\
		& \geq  \indicator_{\{l(1+c_0)\ltwo{\x}\leq y_i\leq u(1-c_0)\ltwo{\x}\}} - \indicator_{\{l\ltwo{\x}\leq y_i\leq u\ltwo{\x}\}}\\
		&= \left\{\begin{array}{rlc}
			&-1, &\text{~if~}y_i\in[l\ltwo{\x}, l(1+c_0)\ltwo{\x}]\cup[u(1-c_0)\ltwo{\x}, u\ltwo{\x}],\\
			&0, &\text{~otherwise.}
		\end{array}
		\right.
	\end{aligned}
\end{equation*}
Therefore
\begin{align}\label{upperbound:indicator}
	\bigg| \indicator_{\{l\nu\leq y_i\leq u\nu\}} - \indicator_{\{l\ltwo{\x}\leq y_i\leq u\ltwo{\x}\}}\bigg| \leq \indicator_{\{l(1-c_0)\ltwo{\x}\leq y_i\leq l(1+c_0)\ltwo{\x}\}} + \indicator_{\{u(1-c_0)\ltwo{\x}\leq y_i\leq u(1+c_0)\ltwo{\x}\}}.
\end{align}
Denote the set of \blue{normalized} $s$-sparse vectors in $\RR^n$ by $\TT_{s}^n = \{\w\in\RR^n:\ltwo{\w} = 1, \lzero{\w} \leq s\}$.
Then for all $\delta\in(0,1/2)$, there exists a set $\calN_{\delta}$ such that for all $\w\in\TT_s^n$, there exists $\wt\w\in\calN_{\delta}$, $\text{supp}(\w) = \text{supp}(\wt\w)$ and $\ltwo{\w-\wt\w}\leq \delta$ and $|\calN_{\delta}|\leq \binom{n}{s} (\frac{3}{\delta})^s\leq (\frac{3en}{\delta s})^s$ \blue{\cite{liu2021towards}}.
Then from the definition of $\rho(\E_1,s)$, we see
\begin{align*}
	\rho(\E_1,s) &= \max_{\w\in\TT_{s}^n}\bigg|\w^T\frac{1}{m}\sum_{i=1}^my_i^2\a_i\a_i^T\bigg(\indicator_{\{l\nu\leq y_i\leq u\nu\}} - \indicator_{\{l\ltwo{\x}\leq y_i\leq u\ltwo{\x}\}}\bigg)\w\bigg|\\
	&\leq \max_{\w\in\TT_{s}^n}\frac{1}{m} \sum_{i=1}^my_i^2(\w^T\a_i)^2\bigg( \indicator_{\{l(1-c_0)\ltwo{\x}\leq y_i\leq l(1+c_0)\ltwo{\x}\}} + \indicator_{\{u(1-c_0)\ltwo{\x}\leq y_i\leq u(1+c_0)\ltwo{\x}\}}\bigg),
\end{align*}
where the last inequality follows form \eqref{upperbound:indicator}.
Denote
\begin{align*}
	\A_1 = \frac{1}{m} \sum_{i=1}^my_i^2\a_i\a_i^T \indicator_{\{l(1-c_0)\ltwo{\x}\leq y_i\leq l(1+c_0)\ltwo{\x}\}},\\
	\A_2 =  \frac{1}{m} \sum_{i=1}^my_i^2\a_i\a_i^T \indicator_{\{u(1-c_0)\ltwo{\x}\leq y_i\leq u(1+c_0)\ltwo{\x}\}}.
\end{align*}
We first consider
$$\rho_1 :=\max_{\w\in\TT_{s}^n}\w^T\A_1\w.$$
Since $\A_2$ is SPD, $\rho_1$ is larger than or equal to the largest eigenvalue of any $s\times s$ principal sub-matrix of $\A_2$
Suppose $\rho_1 = \w_0^T\A\w_0$ for some $\w_0\in\TT_{s}^n$. From the definition of $\calN_{\delta}$, there exists $\w\in\calN_{\delta}$, such that $\supp(\w_0) = \supp(\w)$ and $\ltwo{\w-\w_0}\leq \delta$,
then
\begin{align*}
	\w_0^T\A_1\w_0  = (\w-\w_0)^T\A_1\w + \w_0^T\A(\w-\w_0) + \w^T\A_1\w \leq 2\delta \rho_1 +  \w^T\A_1\w.
\end{align*}
This implies
\begin{align}\label{eq:rho1}
	\rho_1\leq(1-2\delta)^{-1}\max_{\w\in\calN_{\delta}} \w^T\A_1\w.
\end{align}
Now for any $\w\in\calN_{\delta}$, we consider $\w^T\A\w$. Denote $$Z_i = \w^T(y_i^2\a_i\a_i^T \indicator_{\{l(1-c_0)\ltwo{\x}\leq y_i\leq l(1+c_0)\ltwo{\x}\}})\w. $$
Then $\w^T\A_1\w = \frac{1}{m}\sum_{i=1}^mZ_i$. Then there exists some absolute constant $C_1>0$,
\begin{align*}
	\|Z_i\|_{\psi_1} \leq C_1l^2\ltwo{\x}^2.
\end{align*}
Using the centering of $\psi_1$ norm \cite[Exercise 2.7.10]{vershynin2018high}, we see
\begin{align*}
	\|Z_i-\EE Z_i\|_{\psi_1} \leq C_2l^2\ltwo{\x}^2.
\end{align*}
Here $\EE Z_i = (\gamma_4-\gamma_2)(\w^T\x)^2 + \gamma_2\ltwo{\x}^2 \leq \gamma_4\ltwo{\x}^2$ with $\gamma_k = \EE g^k\indicator_{\{l(1-c_0)\leq |g|\leq l(1+c_0)\}}$ for $g\sim N(0,1)$ from Lemma \ref{lemma:expectation}.
And
\begin{align}\label{bound:gamma4}
	\gamma_4 = 2\int_{l(1-c_0)}^{l(1+c_0)}\frac{1}{\sqrt{2\pi}}e^{-x^2/2}x^4dx \leq C_3\sqrt{\frac{\log(mn)}{m}}.
\end{align}

Using Bernstein's inequality (c.f. \cite[Theorem 2.8.1]{vershynin2018high}), we conclude
\begin{align*}
	\PP\bigg(\frac{1}{m}\sum_{i=1}^mZ_i\geq t\ltwo{\x}^2 + \gamma_4\ltwo{\x}^2\bigg)\leq \exp(-c_1\min\{\frac{mt^2}{l^4}, \frac{mt}{l^2}\}).
\end{align*}
Taking union bound over all $\w\in\calN_{\delta}$, we see with probability exceeding $1- (\frac{9en}{\delta^2s})^s\exp(-c_1\min\{\frac{mt^2}{l^4}, \frac{mt}{l^2}\})$,
$$\max_{\w\in\calN_{\delta}} \w^T\A_1\w \leq  t\ltwo{\x}^2 + \gamma_4\ltwo{\x}^2\leq (C_3\sqrt{\frac{\log(mn)}{m}} + t)\ltwo{\x}^2,$$
where the last inequality is from \eqref{bound:gamma4}. By setting $\delta = 1/2$, together with \eqref{eq:rho1}, we obtain with probability exceeding $1 - (36ens^{-1})^{s}\exp(-c_1\min\{\frac{mt^2}{l^4}, \frac{mt}{l^2}\})$,
\begin{align}\label{boundforrho1}
	\rho_1 \leq 2(C_3\sqrt{\frac{\log(mn)}{m}} + t)\ltwo{\x}^2.
\end{align}
Similarly we can show $\rho_2:= \max_{\w\in\TT_{s}^n}\w^T\A_2\w$ has the following upper bound
\begin{align}\label{boundforrho2}
	\rho_2 \leq 2(C_3\sqrt{\frac{\log(mn)}{m}} + t)\ltwo{\x}^2.
\end{align}
with probability exceeding $1 - (36ens^{-1})^{s}\exp(-c_1\min\{\frac{mt^2}{u^4}, \frac{mt}{u^2}\})$.
And \eqref{boundforrho1}, \eqref{boundforrho2} imply $$\rho(\E_1,s)\leq 4(C_3\sqrt{\frac{\log(mn)}{m}} + t)\ltwo{\x}^2.$$
Next we consider $\rho(\E_2,s)$. Recall
$$\E_2 =  \frac{1}{m}\sum_{i=1}^m\bigg(y_i^2\a_i\a_i^T\indicator_{\{l\ltwo{\x}\leq y_i\leq u\ltwo{\x}\}}  - \EE \bar \Y_0\bigg).$$
From the definition, $\rho(\E_2,s)$ is the largest operator norm of any $s\times s$ principal sub-matrix of $\E_2$, and thus
$$\rho(\E_2,s) =\max_{\substack{\ltwo{\z},\ltwo{\y}\leq 1\\ \lzero{\z},\lzero{\y}\leq s\\ \supp(\z) = \supp(\y)}}\z^T\E_2\y.$$
Let $\rho(\E_2,s) = \z_0\E_2\y_0$ for some $\ltwo{\z_0},\ltwo{\y_0}\leq 1, \lzero{\z_0},\lzero{\y_0}\leq s$. Then from the definition of $\calN_{\delta}$, there exists $\z,\y\in\calN_{\delta}$ s.t. $\ltwo{\z-\z_0},\ltwo{\y-\y_0}\leq \delta$ and $\supp(\z) = \supp(\y) = \supp(\z_0)$. Moreover,
$$\rho(\E_2,s) = \z_0^T\E_2\y_0  = (\z_0-\z)^T \E_2 \y_0 + \z^T\E_2(\y_0-\y)+\z^T\E_2\y\leq 2\delta \rho(\E_2,s) +\z^T\E_2\y.$$
Therefore,
$$
\rho(\E_2,s) \leq (1-2\delta)^{-1}\cdot\max_{\substack{\z,\y\in\calN_{\delta}\\ \supp(\z) = \supp(\y)}}\z^T\E_2\y.
$$
Now for any $(\z,\y)\in\calN_{\delta}\times \calN_{\delta}$, we provide the upper bound of  $\z^T\E_2\y$ as follows. Denote $Y_i = \z^T(y_i^2\indicator_{\{l\ltwo{\x}\leq y_i\leq u\ltwo{\x}\}}\a_i\a_i^T - \EE\bar\Y_0)\y$, and then $\z^T\E_2\y = \frac{1}{m}\sum_{i=1}^mY_i$. Then we have
$$
\|Y_i\|_{\psi_1} \leq C_4\left\|\indicator_{\{l\ltwo{\x}\leq y_i\leq u\ltwo{\x}\}}y_i^2\z^T\a_i\a_i^T\y\right\|_{\psi_1}\leq C_5u^2\ltwo{\x}^2
$$
Using Bernstein's inequality (c.f. \cite[Theorem 2.8.1]{vershynin2018high}), we get
$$
\PP(\z^T\E_2\y \geq t\ltwo{\x}^2) \leq \exp\left(-c_2\min\left\{\frac{mt^2}{u^4},\frac{mt}{u^2}\right\}\right).
$$
By taking union bound and setting $\delta = 1/4$, we obtain that, with probability at least $1 - (\frac{36en}{s})^{2s}\exp(-c_2\min\{\frac{mt^2}{u^4},\frac{mt}{u^2}\})$,
it holds $\rho(\E_2,s) \leq 2t\ltwo{\x}^2.$
Putting everything together, we conclude as long as $m\geq C_6s\log n\max\{\frac{u^4}{t^2},\frac{u^2}{t}\}$, with probability exceeding $1-n^{-s}-\frac{2}{mn}$,
$$\rho(\E,s)\leq \rho(\E_1,s) + \rho(\E_2,s) \leq (C_3\sqrt{\frac{\log(mn)}{m}} + 6t)\ltwo{\x}^2\leq 7t\ltwo{\x}^2.$$
\end{proof}
}

Recall $\ybar_0 = \frac{1}{m}\sum_{i=1}^m y_i^2\a_i\a_i^T\indicator_{\{l\leq (y_i/\ltwo{\x})\leq u\}}$ has expectation $\EE \ybar_0 = (\beta-\alpha)\x\x^T + \alpha\ltwo{\x}^2\I$.
Now we can bound the second-to-largest eigenvalue in absolute value of a sub-matrix $\ybar_{\Lambda}$ using the following lemma.

\begin{lemma}\label{lemma:barx-x}
	Let $\Lambda\subset[n]$ be such that $\Lambda\cap\calS\neq\emptyset$, $|\Lambda| = k$, and $\ltwo{[\x^0]_{\Lambda}}\geq \frac{\sqrt{3}}{2}$, where $\x^0=\x/\ltwo{\x}$. Let $\bar\x_{\Lambda}$ be an eigenvector of unit length corresponding to the largest eigenvalue of $\ybar_{\Lambda}$. If $\rho(\E,k) \leq \frac{1}{4}(\frac{3}{4}\beta - \alpha)\ltwo{\x}^2$, then we have
	$$\dist(\bar\x_{\Lambda}, [\x^0]_{\Lambda})^2\leq \ltwo{[\x^0]_{\Lambda}}^2 + 1-2\frac{\ltwo{[\x^0]_{\Lambda}}}{\sqrt{1+\blue{\frac{4\rho^2(\E,k)}{(\frac{3}{4}\beta-\alpha)^2\ltwo{\x}^4}}}}.$$
\end{lemma}
\begin{proof}
	Recall $\E = \ybar - \EE\ybar_0$. Denote $\bar\lambda$ the largest eigenvalue of $\ybar_{\Lambda}$.
	\blue{From Weyl's inequality \cite[Theorem 4.3.1]{horn2012matrix}, we have}
%	From the perturbation theory of the symmetric eigenvalue problem, we have
	\begin{align}\label{lowerbound:l1}
		\bar\lambda=\lambda_1(\ybar_{\Lambda})&\geq \lambda_1([\EE\ybar_0]_{\Lambda}) + \lambda_n(\E_{\Lambda})\geq \blue{(\beta-\alpha)}\ltwo{\x_{\Lambda}}^2+\blue{\alpha}\ltwo{\x}^2 - \rho(\E,k)\notag\\
		&\blue{\geq \beta\ltwo{\x_{\Lambda}}^2- \rho(\E,k) = \beta\ltwo{[\x^0]_{\Lambda}}^2\ltwo{\x}^2 - \rho(\E,k)}\notag\\
        &\geq \frac{3}{4}\beta\ltwo{\x}^2- \rho(\E,k),
	\end{align}
	where the last inequality holds since $\ltwo{[\x^0]_{\Lambda}}\geq \frac{\sqrt{3}}{2}$ and $\lambda_n(\E_{\Lambda})\geq - \rho(\E,k)$ by definition. Similarly, we have for all $j\geq 2$,
	\begin{align}\label{upperbound:lj}
		|\lambda_j(\ybar_{\Lambda})| \leq |\lambda_j([\EE\ybar_0]_{\Lambda})| + \rho(\E,k)\leq \alpha\ltwo{\x}^2 + \rho(\E,k).
	\end{align}
	Notice $\ltwo{\bar\x_{\Lambda}} = 1$ but $\ltwo{[\x^0]_{\Lambda}}\leq 1$.
	Now we write $\bar\x_{\Lambda} = c_1[\x^0]_{\Lambda}/\ltwo{[\x^0]_{\Lambda}} + c_2\z$ with $([\x^0]_{\Lambda})^T\z=0$, $\ltwo{\z} = 1$ and $c_1^2 + c_2^2 = 1$.
	As a quick consequence, $\supp(\z)\subset \Lambda$.
	Then we have
	\begin{align*}
		\bar\lambda c_1[\x^0]_{\Lambda}/\ltwo{[\x^0]_{\Lambda}} +
		\bar\lambda c_2\z = \bar\lambda \bar\x_{\Lambda} = \ybar_{\Lambda}\bar\x_{\Lambda} =  c_1\ybar_{\Lambda}[\x^0]_{\Lambda}/\ltwo{[\x^0]_{\Lambda}} + c_2\ybar_{\Lambda}\z.
	\end{align*}
	By taking the inner product with $\z$,  we obtain
	$$\bar\lambda c_2 = c_1\z^T\ybar_{\Lambda}[\x^0]_{\Lambda}/\ltwo{[\x^0]_{\Lambda}} + c_2\z^T\ybar_{\Lambda}\z.$$
	\blue{Since $[\x^0]_{\Lambda}$ is the eigenvector of $[\EE \bar\Y_0]_{\Lambda}$, we have $\z^T[\EE \bar\Y_0]_{\Lambda}[\x^0]_{\Lambda} = 0$. This leads to}
	$$
    |c_2| = |c_1|\frac{|\z^T([\EE\bar\Y_0]_{\Lambda} + \E_{\Lambda})[\x^0]_{\Lambda}|/\ltwo{[\x^0]_{\Lambda}}}{|\bar\lambda - \z^T\ybar_{\Lambda}\z|}= |c_1|\underbrace{\frac{|\z^T\E_{\Lambda}[\x^0]_{\Lambda}|/\ltwo{[\x^0]_{\Lambda}}}{|\bar\lambda - \z^T\ybar_{\Lambda}\z|}}_{:=t}.
    $$
    \blue{Since $\supp(\z)\subset\Lambda$, we have $|\z^T\E_{\Lambda}[\x^0]_{\Lambda}|/\ltwo{[\x^0]_{\Lambda}} \leq \rho(\E,k).$
    	Moreover, since $\z$ is perpendicular to $[\x^0]_{\Lambda}$, which corresponds to the largest eigenvalue of $\bar\Y_{\Lambda}$, we have from \eqref{upperbound:lj}
    	$$|\z^T\bar\Y_{\Lambda}\z| \leq \max_{j\neq 1}\lambda_j(\bar\Y_{\Lambda})\leq \alpha\ltwo{\x}^2  + \rho(\E,k).$$
So we have from \eqref{lowerbound:l1} together with $\rho(\E,k) \leq \frac{1}{4}(\frac{3}{4}\beta - \alpha)\ltwo{\x}^2$,
\begin{align*}
	 t\leq \frac{\rho(\E,k)}{\bar\lambda - |\z^T\bar\Y_{\Lambda}\z|} \leq \frac{2\rho(\E,k)}{(\frac{3}{4}\beta-\alpha)\ltwo{\x}^2}.
\end{align*}
}

	Then, we have $1 = c_1^2+c_2^2\leq (1+t^2)c_1^2$, which implies $c_1^2\geq \frac{1}{1+t^2}$.
	Now, without loss of generality, we assume $c_1>0$. So,
	\begin{align*}
	\ltwo{[\x^0]_{\Lambda} -\bar\x_{\Lambda}}^2
	= \ltwo{[\x^0]_{\Lambda}}^2+1 -2([\x^0]_{\Lambda})^T\bar\x_{\Lambda}
	&= \ltwo{[\x^0]_{\Lambda}}^2+1 -2c_1\ltwo{[\x^0]_{\Lambda}}\cr
	&\leq \ltwo{[\x^0]_{\Lambda}}^2+1 -2\frac{\ltwo{[\x^0]_{\Lambda}}}{\sqrt{1+t^2}}.
	\end{align*}
	Finally, we plug the upper bound for $t$ and get the desired result.
\end{proof}

\subsection{Proof of Theorem \ref{thm:onestep}}\label{proof:thm:onestep}
We are ready to prove \cref{thm:onestep}. We organize the proof into three parts. Firstly, we show $j_0$ chosen in Algorithm \ref{alg:init:onestep} satisfies $|x_{j_0}^0|\geq\frac{\linf{\x^0}}{2}$. Secondly, we show that $\ltwo{[\x^0]_{\hat\calS}}$ \blue{is sufficient large, meaning that $\hat\calS$ captures the indices of $\x$ with largest absolute values.}
Lastly, we put everything together.

\vspace{0.5cm}

\noindent\textit{Step 1: Estimating $|x_{j_0}^0|$.} We start with the estimation of $|x_{j_0}^0|$.
\begin{lemma}\label{lemma:xj0}
	If the sample size $m$ satisfies $m\geq C\linf{\x^0}^{-4}\log n$ for some absolute constant $C>0$, then, with probability exceeding $1 - \calO(n^{-10})$, $j_0$ defined by $j_0 = \arg\max_{1\leq j\leq n} Y_{jj}$ satisfies $|x_{j_0}^0|\geq\frac{\linf{\x^0}}{2}$.
\end{lemma}
\begin{proof}
Let $j^*\in[n]$ be satisfying $|x_{j^*}^0| = \linf{\x^0}$. Then \cite[Lemma 1]{wang2017sparse} implies that, for any $\epsilon_1>0$, with probability exceeding $1-\exp(-\frac{m\epsilon_1^2}{192})$,
	\begin{align}\label{Yjj:lowerbound}
		Y_{j^*j^*}\ltwo{\x}^{-2}\geq 1+ 2\linf{\x^0}^2 - \epsilon_1.
	\end{align}
On the other hand, we consider the set $\calS_1: = \{j\in[n]: |x^0_j|\leq \frac{\linf{\x^0}}{2}\}$. By using \cite[Lemma 1]{wang2017sparse} and taking union bound, we have that, with probability at least $1 - (n-1)\exp(-\frac{m\epsilon_2^2}{192})$,
	\begin{align}\label{Yjj:upperbound}
		Y_{jj}\ltwo{\x}^{-2}\leq 1+ \frac{\linf{\x^0}^2}{2} + \epsilon_2,\qquad \forall~j\in\calS_1.
	\end{align}
Now combine \eqref{Yjj:lowerbound} and \eqref{Yjj:upperbound} and set $\epsilon_1 = \epsilon_2 = \frac{\linf{\x^0}^2}{2}$. Then,
with probability exceeding $1 - n\exp(-\frac{m\linf{\x^0}^4}{768})$, $Y_{j_0j_0}\geq Y_{j^*j^*} > Y_{jj}$ for all $j\in\calS_1$, which implies $j_0\not\in\calS_1$.
\end{proof}

\vspace{0.5cm}

\noindent\textit{Step 2: Estimating $\ltwo{[\x^0]_{\hat\calS}}$.} For any $\gamma\in(0,1]$, we define $\calS^-_{\gamma} :=\{i\in\calS: |x^0_i|<\frac{\gamma}{2\sqrt{s}}\}$ and $\calS^+_{\gamma} := \calS\backslash\calS^-_{\gamma}$. Then $\ltwo{[\x^0]_{\calS_{\gamma}^-}}^2 < \frac{\gamma^2}{4s}s = \frac{\gamma^2}{4}$ and $\ltwo{[\x^0]_{\calS^+_{\gamma}}}^2 > 1-\frac{\gamma^2}{4}$. Since $\linf{\x^0}\geq\frac{1}{\sqrt{s}}$, \cref{lemma:xj0} implies $|x^0_{j_0}|\geq \frac{1}{2\sqrt{s}}\geq\frac{\gamma}{2\sqrt{s}}$ with high probability, and thus $j_0\in\calS^+_{\gamma}$. The following lemma shows that $\calS^+_{\gamma}\subset\hat\calS$ with high probability, where $\hat\calS$ is defined in \eqref{eq:hatS}.

\begin{lemma}\label{lemma:hatOmega}
Assume the number $m$ of samples satisfies
$$
m\geq C\max\{\gamma^{-1}\linf{\x^0}^{-1}\sqrt{s}\log^3n, \gamma^{-2}\linf{\x^0}^{-2}s\log n\}
$$
for some absolute constant $C>0$. Then, with probability exceeding $1-\calO(n^{-10})$, $\hat\calS$ chosen in Algorithms \ref{alg:init} and \ref{alg:init:onestep} satisfies $\calS^+_{\gamma}\subset\hat\calS$.
\end{lemma}
\begin{proof}
It suffices to show
\begin{equation}\label{eq:minS+>maxSc}
\min_{j\in\calS^+_{\gamma}}|[\Y\e_{j_0}]_j| > \max_{j\in\calS^c}|[\Y\e_{j_0}]_j|.
\end{equation}
To this end, we need to bound the expectation $|[\EE\Y\e_{j_0}]_j|$ and the concentration $|[\Y\e_{j_0}]_j - [\EE\Y\e_{j_0}]_j|$ for all $j\in\calS^+_{\gamma}$ and $j\in\calS^c$.

To bound the expectation, we recall $|[\EE\Y\e_{j_0}]_j|$ given in  \eqref{eq:EYj0}.
Then, tor $j\in\calS^+_{\gamma}$, the definition of $\calS^+_{\gamma}$ gives $|[\EE\Y\e_{j_0}]_j|\geq \frac{\gamma |x_{j_0}|\ltwo{\x}}{\sqrt{s}}$. This together with \cref{lemma:xj0} implies
\begin{equation}\label{eq:EYj0Sc}
|[\EE\Y\e_{j_0}]_j|\geq\frac{\gamma\linf{\x}\ltwo{\x}}{2\sqrt{s}},\quad\forall~j\in\calS^+_{\gamma}
\end{equation}
with high probability.

For the concentration, we decompose $[\Y\e_{j_0}]_j = \frac{1}{m}\sum_{i=1}^my_i^2 a_{ij} a_{ij_0}=:\frac{1}{m}\sum_{i=1}^mZ_{ij,1}+\frac{1}{m}\sum_{i=1}^mZ_{ij,2}$, where
$$
Z_{ij,1} = y_i^2 a_{ij} a_{ij_0}\indicator_{\{\max\{|y_i|/\ltwo{\x},| a_{ij}|,|a_{ij_0}|\}\leq \sqrt{44\log n}\}},
\quad
Z_{ij,2} = y_i^2 a_{ij} a_{ij_0} - Z_{ij,1}.
$$
Thus, we have $|Z_{ij,1}|\leq (44\log n)^2\ltwo{\x}^2$ and
$$
	\frac{1}{m^2}\sum_{i=1}^m \var(Z_{ij,1}) \leq \frac{1}{m^2}\sum_{i=1}^m \EE y_i^4 a_{ij}^2 a_{ij_0}^2\leq \frac{1}{m^2}\sum_{i=1}^m \sqrt{\EE y_i^8\cdot \EE a_{ij}^4 a_{ij_0}^4}\leq \frac{105}{m}\ltwo{\x}^4.
$$
Using Bernstein's inequality (c.f. \cite[Theorem 2.8.4]{vershynin2018high}), we obtain
\begin{align}\label{Z:1}
		\PP\left(\left|\frac{1}{m}\sum_{i=1}^m (Z_{ij,1} - \EE Z_{ij,1})\right|\geq \frac{\gamma\linf{\x}\ltwo{\x}}{16\sqrt{s}}\right)\leq 2\exp\left(-\frac{\frac{\gamma^2\linf{\x}^2\ltwo{\x}^2}{256s}}{2(\frac{105}{m}\ltwo{\x}^4+\frac{44^2\log^2n}{m}\ltwo{\x}^2\frac{\gamma\linf{\x}\ltwo{\x}}{48\sqrt{s}})}\right).
	\end{align}
Therefore, as long as $m\geq C\max\{\gamma^{-1}\ltwo{\x}\linf{\x}^{-1}\sqrt{s}\log^3n, \gamma^{-2}\ltwo{\x}^2\linf{\x}^{-2}s\log n\}$ for some absolute constant $C>0$, we have $|\frac{1}{m}\sum_{i=1}^m (Z_{ij,1} - \EE Z_{ij,1})|\leq \frac{\gamma\linf{\x}\ltwo{\x}}{16\sqrt{s}}$ with probability more than $1 - \calO(n^{-11})$. Moreover, we have
\begin{align*}
		\var\left(\frac{1}{m}\sum_{i=1}^{m}Z_{ij,2}\right) &\leq \frac{1}{m}\EE(y_1^4 a_{1j}^2 a_{1j_0}^2\indicator_{\{\max\{|y_1|/\ltwo{\x},|a_{1j}|,|a_{1j_0}|\}\geq \sqrt{44\log n}\}})\no\\
		&\leq \frac{1}{m}\sqrt{\EE y_1^8 a_{1j}^4 a_{1j_0}^4}\cdot\PP^{1/2}\left(\max\{|y_1|/\ltwo{\x},|a_{1j}|,|a_{1j_0}|\}\geq \sqrt{44\log n}\right)\no\\
		&\leq \frac{45\sqrt{1001}}{m}\ltwo{\x}^4\cdot 3n^{-11},
\end{align*}
where the second inequality is from the Cauchy-Schwartz inequality, and the third inequality follows the fact that $|y_1|/\ltwo{\x},|a_{1j}|,|a_{1j_0}|$ are all the absolute value of Gaussian random variables.
From Chebyshev's inequality, we see
\begin{align}\label{Z:2}
		\PP\left(\left|\frac{1}{m}\sum_{i=1}^m (Z_{ij,2} - \EE Z_{ij,2})\right|\geq \frac{\gamma\linf{\x}\ltwo{\x}}{16\sqrt{s}}\right)\leq \frac{\frac{45\sqrt{1001}}{m}\ltwo{\x}^4\cdot 3n^{-11}}{\frac{\gamma^2\linf{\x}^2\ltwo{\x}^2}{256s}}\leq \calO(n^{-11}),
\end{align}
where the last inequality holds as long as $m\geq C\gamma^{-2}\ltwo{\x}^2\linf{\x}^{-2}s$. Putting \eqref{Z:1} and \eqref{Z:2} together and taking a union bound over all $j\in[n]$, we get that, with probability at least $1 - \calO(n^{-10})$,
\begin{equation}\label{eq:concYej0}
	|[\Y\e_{j_0}]_j - [\EE\Y\e_{j_0}]_j| \leq \frac{\gamma\linf{\x}\ltwo{\x}}{8\sqrt{s}},\quad\forall~j\in[n].
\end{equation}

By combining \eqref{eq:EYj0Sc} and \eqref{eq:concYej0}, we obtain that
$$
|[\Y\e_{j_0}]_j|\geq |[\EE\Y\e_{j_0}]_j|-|[\Y\e_{j_0}]_j - [\EE\Y\e_{j_0}]_j|\geq \frac{3\gamma\linf{\x}\ltwo{\x}}{8\sqrt{s}},
\quad\forall~j\in\calS^+_{\gamma}.
$$
Moreover, it follows from \eqref{eq:EYj0} and \eqref{eq:concYej0} that
$$
|[\Y\e_{j_0}]_j|\leq |[\EE\Y\e_{j_0}]_j|+|[\Y\e_{j_0}]_j - [\EE\Y\e_{j_0}]_j|\leq\frac{\gamma\linf{\x}\ltwo{\x}}{8\sqrt{s}},
\quad\forall~j\in\calS^c.
$$
These two inequalities lead to \eqref{eq:minS+>maxSc}, which concludes the proof.
\end{proof}	

	\vspace{0.5cm}
	
	\noindent\textit{Step 3: Putting everything together.} Now we estimate $\dist(\hat\x,\x)$ and prove \cref{thm:onestep}.
 \begin{proof}[Proof of \cref{thm:onestep}]
For simplicity, we denote $\rho = \rho(\E,s)$.
\blue{Also recall we would like to obtain some $\hat \x$ such that $\ltwo{\hat\x-\x}\leq \delta\ltwo{\x}$.}
To proceed, we choose $t = \frac{\delta}{448}$ in \cref{lemma:ssparseeigenvalue} and $\gamma=\frac{\delta}{4}$ in \cref{lemma:hatOmega}.
\blue{Notice $\linf{\x^0} = \linf{\x}/\ltwo{\x} = \stablesparsity{\x}^{-1/2}$. }
By applying \cref{lemma:nu:concentration}, \cref{lemma:ssparseeigenvalue}, \cref{lemma:xj0}, and \cref{lemma:hatOmega}, we obtain that: as long as
%$m\geq C\max\{\delta^{-2}\stablesparsity{\x}s,\delta^{-4}\}\log n$
\blue{$$m\geq C\max\{\delta^{-1}(\stablesparsity{\x}s)^{1/2}\log^3n, \delta^{-2}\stablesparsity{\x}s\log n\}$$}
 for some $C>0$ depending only on $u$, with probability at least $1-n^{-s}-2n^{-10}-2m^{-1}n^{-1}\geq 1 -5m^{-1}n^{-1}$, it holds that
\begin{equation}\label{eq:succthm1}
\rho\leq 7t\ltwo{\x}^2,\quad j_0\in\calS_{\gamma}^+\subset\hat{\calS},\quad \left|\nu^2 - \ltwo{\x}^2 \right|\leq \delta\ltwo{\x}^2.
\end{equation}

We first estimate $\dist(\hat\x^0,\x^0)^2$ under \eqref{eq:succthm1}.  Since $\hat\x_0$ in \cref{alg:init:onestep} is supported only on $\hat\calS$,
\begin{equation}\label{eq:distdecomp}
\dist(\hat\x^0,\x^0)^2 = \dist(\hat\x^0, [\x^0]_{\hat\calS})^2+\ltwo{[\x^0]_{{\hat\calS}^c}}^2.
\end{equation}
From $\calS_{\gamma}^+\subset\hat{\calS}$ in \eqref{eq:succthm1} and the definition of $\calS^-_{\gamma}$, it follows that $\ltwo{[\x^0]_{\hat\calS^c}}^2\leq \ltwo{[\x^0]_{\calS^-_{\gamma}}}^2<\frac{\gamma^2}{4}$. It remains to estimate $\dist(\hat\x^0, [\x^0]_{\hat\calS})^2$. To this end, we apply \cref{lemma:barx-x} with $\Lambda=\hat\calS$ and $k=s$. Since $\delta\in(0,1)$, we have $\ltwo{[\x^0]_{\hat\calS}}^2=1-\ltwo{[\x^0]_{\hat\calS^c}}^2\geq 1-\frac{\gamma^2}{4}= 1-\frac{\delta^2}{64}>\frac{\sqrt3}{2}$. Moreover, due to the assumptions $\beta/\alpha\geq 2$ and $\alpha\geq 1/2$, it holds that $\frac{3}{4}\beta-\alpha=(\frac{3}{4}\frac{\beta}{\alpha}-1)\alpha\geq\frac{1}{4}$, which together with the bound of $\rho$ in \eqref{eq:succthm1} implies $\rho\leq \frac{1}{4}(\frac{3}{4}\beta - \alpha)\ltwo{\x}^2$. Therefore, \cref{lemma:barx-x} gives
$$
 \dist(\hat\x^0,[\x^0]_{\hat\calS})^2\leq \ltwo{[\x^0]_{\hat\calS}}^2 + 1-2\frac{\ltwo{[\x^0]_{\hat\calS}}}{\sqrt{1+\frac{4\rho^2(\E,k)}{(\frac{3}{4}\beta-\alpha)^2\ltwo{\x}^4}}}.
$$
As $\sqrt{1-\frac{\gamma^2}{4}}\leq\ltwo{[\x^0]_{\hat\calS}}\leq 1$, we have
\begin{align*}
	\dist(\hat\x^0,[\x^0]_{\hat\calS})^2 &\leq \max\bigg\{2-2\frac{1}{\sqrt{1+\frac{4\rho^2}{(\frac{3}{4}\beta-\alpha)^2\ltwo{\x}^4}}}, 2-\frac{\gamma^2}{4}-2\frac{\sqrt{1-\frac{\gamma^2}{4}}}{\sqrt{1+\frac{4\rho^2}{(\frac{3}{4}\beta-\alpha)^2\ltwo{\x}^4}}}\bigg\}\\
	&\leq \max\bigg\{\frac{8\rho^2}{(\frac{3}{4}\beta-\alpha)^2\ltwo{\x}^4+4\rho^2}, 2-\frac{\gamma^2}{4}-2\frac{1-\frac{\gamma^2}{4}}{1+\frac{4\rho^2}{(\frac{3}{4}\beta-\alpha)^2\ltwo{\x}^4}}\bigg\}\\
	&\leq \frac{8\rho^2 + \frac{(\frac{3}{4}\beta-\alpha)^2\ltwo{\x}^4}{2}\gamma^2}{(\frac{3}{4}\beta-\alpha)^2\ltwo{\x}^4 +4\rho^2} - \frac{\gamma^2}{4}
	\leq \frac{8\rho^2}{(\frac{3}{4}\beta-\alpha)^2\ltwo{\x}^4 +4\rho^2} + \frac{\gamma^2}{4}\\
   &\leq \frac{8\rho^2}{\frac{1}{16}\ltwo{\x}^4 +4\rho^2} + \frac{\gamma^2}{4}\leq 128\rho^2\ltwo{\x}^{-4}+\frac{\gamma^2}{4},
\end{align*}
where in the first inequality of the fourth line we have used $\frac{3}{4}\beta-\alpha\geq\frac14$, and in the second inequality we have used $1-\frac{1}{\sqrt{1+y}}\leq \frac{y}{y+1}$ for all $y>0$.
Combining it with \eqref{eq:distdecomp}, and recall we set $\gamma = \frac{\delta}{4}$, $t = \delta/448$, we get
	\begin{align}\label{eq:distx0x}
		\dist(\hat\x^0,\x^0)^2 \leq \frac{\gamma^2}{2} + 128\rho^2\ltwo{\x}^{-4}\leq\gamma^2 = \delta^2/16,
	\end{align}
where we have used the bound of $\rho$ in \eqref{eq:succthm1}.

Finally, since we have defined $\hat\x=\nu\hat\x^0$ and $\x^0=\x/\ltwo{\x}$, combining \eqref{eq:distx0x} and the bound of $\nu$ in \eqref{eq:succthm1} gives
\blue{
\begin{equation}\label{eq:disthatxx}
	\begin{split}
		\dist(\hat\x, \x) &\leq \dist(\nu\x^0, \x)+ \dist(\hat\x,\nu\x^0)
		\leq\ltwo{\x-\nu\x^0} + \dist(\nu\x^0 ,\hat\x)\\
		&= |\ltwo{\x} - \nu|  + \nu\dist(\x^0 ,\hat\x^0)\\
		&\leq 2\ltwo{\x}\frac{\delta}{4} + \frac{\delta}{2}\ltwo{\x},
	\end{split}
\end{equation}
where the first inequality is from Lemma \ref{lemma:triangleinequality:dist}.
}
\end{proof}

\subsection{Proof of Theorem \ref{thm:init}}\label{sec:proof:tp}
In this section, we prove \cref{thm:init}. The proof is separated into three parts. Firstly, we show that $\w_0$ falls into a small constant neighborhood of $\x^0:=\x/\ltwo{\x}$. Secondly, we show the convergence of the truncated power method. Thirdly, we show the projection back to $\MM_s$ gives the desired estimator.

\begin{proof}[Proof of \cref{thm:init}]
Throughout the proof, we denote $k = s+2s'$, $\rho = \rho(\E,k)$ with $\E = \ybar-\EE\ybar_0$, $F_t = \supp(\w_t)$, and we set $\gamma=1/4$ in \eqref{eq:distx0x}.
Similar to the proof of Theorem \ref{thm:onestep}, given
%$m\geq C\max\{\stablesparsity{\x}s,\delta^{-2}s,\delta^{-4}\}\log n$
\blue{$$m\geq C\max\{(\stablesparsity{\x}s)^{1/2}\log^3n, \stablesparsity{\x}s\log n, \delta^{-2}s \log n\}$$}
for some constant $C>0$ depending only on $u$,
and setting $t = \frac{1}{70000}\delta$ in Lemma \ref{lemma:ssparseeigenvalue} gives
with probability exceeding $1-n^{-s} - 2n^{-10}-2m^{-1}n^{-1}\geq 1 -5m^{-1}n^{-1}$, the following event holds:
\begin{align}\label{eq:succthm2}
\calE = \left\{|\nu^2 - \ltwo{\x}^2| \leq 3\sqrt{\frac{\log(mn)}{m}}\ltwo{\x}^2,\quad \rho\leq 10^{-4}\delta\ltwo{\x}^2,\quad  \dist(\w_0^0,\x^0)^2\leq 1\right\}.
\end{align}
We shall proceed with the proof conditioning on this event.

\hspace{1cm}

\noindent\textit{Step 1: Estimation of $|(\w_0^0)^T\x_0|$.}
From $\dist(\w_0^0,\x^0)^2\leq 1$, we obtain $|(\w_0^0)^T\x_0|\geq 1/2$.

\hspace{1cm}

\noindent\textit{Step 2: Convergence of truncated power method.} We prove in this step, starting from $\w_0^0$, the truncated power method will output some $\w_{t_{\max}}^0$ such that $\dist(\w_{t_{\max}}^0, \x^0)\leq \delta/2$.

We show $\dist(\w_{t}^0,\x^0)\leq 1$ by induction and estimate $\dist(\w_{t}^0,\x^0)$ in terms of $\dist(\w_{t-1}^0,\x^0)$. The result from Step 1 tells us $\dist(\w_0^0,\x^0)\leq 1$. Suppose we have $\dist(\w_{t-1}^0,\x^0)\leq 1$. We would like to show $\dist(\w_{t}^0,\x^0)\leq 1$ and give a tight bound of $\dist(\w_{t}^0,\x^0)$.

Denote $\Lambda = F_{t-1} \cup F_t\cup\calS$. Then $|\Lambda|\leq 2s'+s = k$.
recall that $\w_t=T_{s'}(\ybar_{\Lambda}\w_{t-1})/\ltwo{T_{s'}(\ybar_{\Lambda}\w_{t-1})}$, and
Define
\begin{align}\label{tildewtprime}
	\wt\w_t^0 = \ybar_{\Lambda}\w_{t-1}^0/\ltwo{\ybar_{\Lambda}\w_{t-1}^0}.
\end{align}
Then $\w_t^0 = [\tilde\w_{t}^0]_{F_t}/\ltwo{[\tilde\w_{t}^0]_{F_t}}$.
Let $\kappa$ be the ratio of the second largest (in absolute value) to the largest eigenvalue of $\ybar_{\Lambda}$. Then, since $\calS\subset\Lambda$,
\begin{align*}
	\kappa=\frac{|\lambda_2(\ybar_{\Lambda})|}{\lambda_1(\ybar_{\Lambda})}\leq \frac{|\lambda_2([\EE\ybar_0]_{\Lambda})| + \rho}{\lambda_1([\EE\ybar_0]_{\Lambda})- \rho} = \frac{\alpha\ltwo{\x}^2+\rho}{\beta\ltwo{\x}^2 - \rho}\leq 0.51,
\end{align*}
where the last inequality follows from $\rho\leq c_0\ltwo{\x}$ and $\beta/\alpha\geq 2$.

Let $\bar\x^0$ be a unit eigenvector corresponding to the top eigenvalue of $\ybar_{\Lambda}$ and satisfying $(\x^0)^T\bar\x^0 \geq 0$. As a result, $\dist(\x^0,\bar\x^0) = \ltwo{\x^0-\bar\x^0}$. Then, since $\wt\w_t^0 = \ybar_{\Lambda}\w_{t-1}^0/\ltwo{\ybar_{\Lambda}\w_{t-1}^0}$, [\cite{yuan2013truncated}, Lemma 11] gives
\begin{align*}
	|(\bar\x^0)^T\tilde\w_{t}^0| &\geq |(\bar\x^0)^T\w_{t-1}^0|\big(1+(1-\kappa^2)(1-((\bar\x^0)^T\w_{t-1}^0)^2)/2\big)\\
	&\geq |(\bar\x^0)^T\w_{t-1}^0|\big(1+0.369(1-((\bar\x^0)^T\w_{t-1}^0)^2)\big),
\end{align*}
which leads to
\begin{equation}\label{eq:keyineq1Thm2}
    1-|(\bar\x^0)^T\tilde\w_{t}^0| \leq (1-|(\bar\x^0)^T\w_{t-1}^0|)\big(1-0.369(|(\bar\x^0)^T\w_{t-1}^0|+|(\bar\x^0)^T\w_{t-1}^0|^2)\big).
\end{equation}
Since $\calS\subset\Lambda$, it is deducted from Lemma \ref{lemma:barx-x} that
\begin{equation}\label{eq:keyineq2Thm2}
\ltwo{\x^0-\bar\x^0}^2 = \dist(\bar\x^0, \x^0)^2\leq 2 - 2\frac{1}{\sqrt{1 + \frac{4\rho^2}{(\frac{3}{4}\beta - \alpha)^2\ltwo{\x}^4}}} \leq 64\rho^2\ltwo{\x}^{-4}\leq 10^{-6}\delta^2,
\end{equation}
where in the last second inequality we used $\frac{3}{4}\beta - \alpha\geq \frac{1}{4}$ and $1-\frac{1}{\sqrt{1+t}}\leq \frac{t}{2}$ for $t\geq 0$, and in the last inequality we used $\rho\ltwo{\x}^{-2} \leq c_0\delta$ and $c_0\leq 0.0001$. Note that $\dist(\x^0,\w_{t-1}^0)\leq 1$ implies $|(\x^0)^T\w_{t-1}^0|\geq 0.5$, which further leads to
\begin{equation*}
\begin{split}
|(\bar\x^0)^T\w_{t-1}^0|&\geq|(\x^0)^T\w_{t-1}^0|-|(\bar\x^0-\x^0)^T\w_{t-1}^0| \geq|(\x^0)^T\w_{t-1}^0|-\ltwo{\bar\x^0-\x^0}\ltwo{\w_{t-1}^0}\cr
&\geq 0.5-0.001\delta\geq 0.499,
\end{split}
\end{equation*}
This plugged into \eqref{eq:keyineq1Thm2} gives
$$
1-|(\bar\x^0)^T\tilde\w_{t}^0| \leq 0.724(1-|(\bar\x^0)^T\w_{t-1}^0|),
$$
which also implies $|(\bar\x^0)^T\tilde\w_{t}^0|\geq 0.276$. Because $\bar\x^0, \tilde\w_t^0, \w_{t-1}^0$ are of unit length, the inequality above is equivalent to
$$
\dist(\bar\x^0,\tilde\w_{t}^0)\leq 0.851\cdot\dist(\bar\x^0,\w_{t-1}^0).
$$
Thus, setting $\u_1 = \x^0, \u_2 = \tilde\w_t^0,\u_3 = \bar\x^0$ and noticing $|(\x^0)^T\tilde\w_t^0|\geq |(\bar\x^0)^T\tilde\w_t^0| - \ltwo{\bar\x^0-\x^0}\geq 0.275>0.001>\ltwo{\bar\x^0-\x^0}$,  \cref{lemma:triangleinequality:dist} derives
$$\dist(\x^0,\tilde\w_t^0)\leq \dist(\bar\x^0,\tilde\w_t^0) + \dist(\bar\x^0,\x^0).$$
Also, since $|(\bar\x^0)^T\w_{t-1}^0|\geq |(\x^0)^T\w_{t-1}^0| - \ltwo{\bar\x^0-\x^0}\geq 0.499>0.001>\ltwo{\bar\x^0-\x^0}$, using \cref{lemma:triangleinequality:dist} again with $\u_1 = \bar\x^0, \u_2 = \w_{t-1}^0,\u_3 = \x^0$ leads to
$$\dist(\bar\x^0,\w_{t-1}^0)\leq \dist(\x^0,\w_{t-1}^0) + \dist(\x^0,\bar\x^0).$$
The three inequalities above together with \eqref{eq:keyineq2Thm2} imply
\begin{equation}\label{eq:distx0tildewt}
\dist(\x^0,\tilde\w_{t}^0)\leq 0.851\cdot\dist(\x^0, \w_{t-1}^0) +0.002\delta. %+ 8\sqrt{6}\rho\ltwo{\x}^{-2}.
\end{equation}
Also, since %$|(\x^0)^T\tilde\w_t^0|\geq 0.276$ and
we set $s' = 50s$, using [\cite{yuan2013truncated}, Lemma 12] and the fact that if $a\geq b - \min\{c_1,c_2\}$, then $a\geq b - c_2$, we have
\begin{equation*}
\begin{split}
|(\x^0)^T[\tilde\w_t^0]_{F_t}| &\geq |(\x^0)^T\tilde\w_t^0| - (s/s')^{1/2}(1+(s/s')^{1/2})(1-|(\x^0)^T\tilde\w_t^0|^2)\cr
&\geq |(\x^0)^T\tilde\w_t^0| - 0.162(1-|(\x^0)^T\tilde\w_t^0|^2),
\end{split}
\end{equation*}
from which it follows that
\begin{equation*}
\begin{split}
1-|(\x^0)^T[\tilde\w_t^0]_{F_t}|&\leq (1-|(\x^0)^T\tilde\w_t^0|)+0.162(1+|(\x^0)^T\tilde\w_t^0|)(1-|(\x^0)^T\tilde\w_t^0|)\cr
&\leq1.324(1-|(\x^0)^T\tilde\w_t^0|),
\end{split}
\end{equation*}
\blue{where the last inequality holds since $|(\x^0)^T\tilde\w_t^0|\leq 1$.}
Recall $\w_t^0 = [\tilde\w_{t}^0]_{F_t}/\ltwo{[\tilde\w_{t}^0]_{F_t}}$. Thus,
\begin{equation}\label{eq:distx0wt0}
\begin{split}
	\dist(\x^0,\w_t^0) &= \sqrt{2(1-|(\x^0)^T\w_t^0|)}=\sqrt{2(1-|(\x^0)^T[\tilde\w_t^0]_{F_t}|/\ltwo{[\tilde\w_t^0]_{F_t}})}\cr
	&\leq \sqrt{2(1-|(\x^0)^T[\tilde\w_t^0]_{F_t}|)}\leq \sqrt{1.324\cdot 2(1 - |(\x^0)^T\tilde\w_t^0|))}\cr
    &=\sqrt{1.324}\cdot\dist(\x^0,\tilde\w_t^0)%\leq 0.97\cdot\dist(\x^0, \w_{t-1}^0) + 22\rho\ltwo{\x}^{-2}\cr
    \leq 0.98\cdot\dist(\x^0, \w_{t-1}^0) + 0.0024\delta,
\end{split}
\end{equation}
where in the last inequality we used \eqref{eq:distx0tildewt}. The above inequality also implies $\dist(\x^0,\w_t^0)\leq 1$. So, we finish the induction, showing $\dist(\x^0,\w_t^0)\leq 1$ for all $t$. Consequently, \eqref{eq:distx0wt0} holds for all $t$, which gives
\begin{equation*}
\begin{split}
\dist(\x^0,\w_t^0)&\leq 0.98\cdot\dist(\x^0, \w_{t-1}^0) + 0.0024\delta \cr
&\leq 0.98^2\cdot\dist(\x^0, \w_{t-2}^0)+0.98\cdot 0.0024\delta+0.0024\delta\cr
&\leq\cdots\cr
&\leq 0.98^t\cdot\dist(\x^0, \w_{0}^0)+\frac{0.0024\delta}{1-0.98}\leq 0.98^t+0.12\delta
\end{split}
\end{equation*}
Therefore, when $t\geq t_{\max}:=\frac{\log(200\delta^{-1})}{\log(1/0.98)}$, we have $0.98^t\leq 0.005\delta$ and hence $\dist(\w_{t}^0,\x^0)\leq \frac{\delta}{8}$.

\hspace{1cm}

\noindent{\textit{Step 3:}} In view of the optimality of $\hat\x^0$ as the closest point to $\w_{t_{\max}}^0$ in $\MM_s$, we have
\begin{align*}
    \ltwo{\hat\x^0-\w_{t_{\max}}^0} \leq \dist(\x^0,\w_{t_{\max}}^0).
\end{align*}
Using triangle inequality for $\dist$ (c.f. Lemma \ref{lemma:triangleinequality:dist}),
\blue{
\begin{align*}
	\dist(\hat\x^0,\x^0) &\leq  \dist(\hat\x^0,\w_{t_{\max}}^0) + \dist(\w_{t_{\max}}^0,\x^0) \\
	&\leq \ltwo{\hat\x^0 - \w_{t_{\max}}^0} +  \dist(\w_{t_{\max}}^0,\x^0) \\
	&\leq 2\dist(\w_{t_{\max}}^0, \x^0)\leq \delta/4.
\end{align*}
}
Now using the same argument as in \eqref{eq:disthatxx} in the proof of Theorem \ref{thm:onestep}, we show $\dist(\hat\x,\x)\leq \delta\ltwo{\x}$ under the event $\calE$. This finishes the proof of the theorem.
\end{proof}

\section{Numerical Experiments}\label{section:experiments}
This section presents some numerical experiment results of the proposed algorithms and compares them with other existing initialization methods. For the refinement step, we use HTP proposed in \cite{cai2022sparse}.

Throughout the experiments, the target signal $\x\in\RR^n$ is $s$-sparse with its support uniformly sampled from all subsets of $[n]$ with size $s$. The non-zero entries of $\x$ are randomly drawn from the mean-$0$ and variance-$1$ Gaussian distribution independently. The sampling vectors $\{\a_i\}$ are i.i.d. standard Gaussian random vectors. The measurements are $\{y_i\}_{i=1}^m$ with $y_i = |\a_i^T\x|$ for $i\in[m]$.

We compare our methods with the spectral initialization \cite{jagatap2019sample}. We shall refer to the initialization method in \cite{jagatap2019sample} as Spectral, Algorithm \ref{alg:init:onestep} as Modified Spectral, and Algorithm \ref{alg:init} as TP.  All the tests are run on a laptop with a 2.4 GHz octa-core i9 processor and 32 GB memory using MATLAB R2020b. Denote the output of the algorithm by $\hat\x$, and the relative error is defined as
$$r(\hat\x,\x) = \frac{\ltwo{\hat\x-\x}}{\ltwo{\x}}.$$
In the experiments, we consider an instance to be successful if $\hat\x$ satisfies $r(\hat\x,\x)\leq 0.001$.

In the first experiment, we fix the dimension $n = 1000$ and choose $s\in\{25,35\}$. For each $s$, we vary the sample size $m$ from $100$ to $1500$. We compare our algorithms with the spectral method. Since we are only interested in the initialization behavior, we shall fix the algorithm in the refinement stage to be HTP \cite{cai2022sparse} as described in \cref{alg:htp}. The results are displayed in Figure \ref{fig:fixs}.

From Figure \ref{fig:fixs}, we can see the TP has overall the best performance. In particular, when the sample size is insufficient, both Modified Spectral and TP have a higher chance of recovering the signal. Moreover, in both Modified Spectral and TP, the multiple-restarted version significantly increases the likelihood of successful recovery. The ability to insert the multiple-restarted step can be also regarded as an advantage of Modified Spectral and TP over Spectral. It is also clearly seen from Figure \ref{fig:fixs} that TP outperforms Modified Spectral with or without multiple restarts, which is also consistent with our theoretical result.

\begin{figure}
 	\centering
	\begin{subfigure}[b]{.98\linewidth}
			\includegraphics[width=0.45\textwidth]{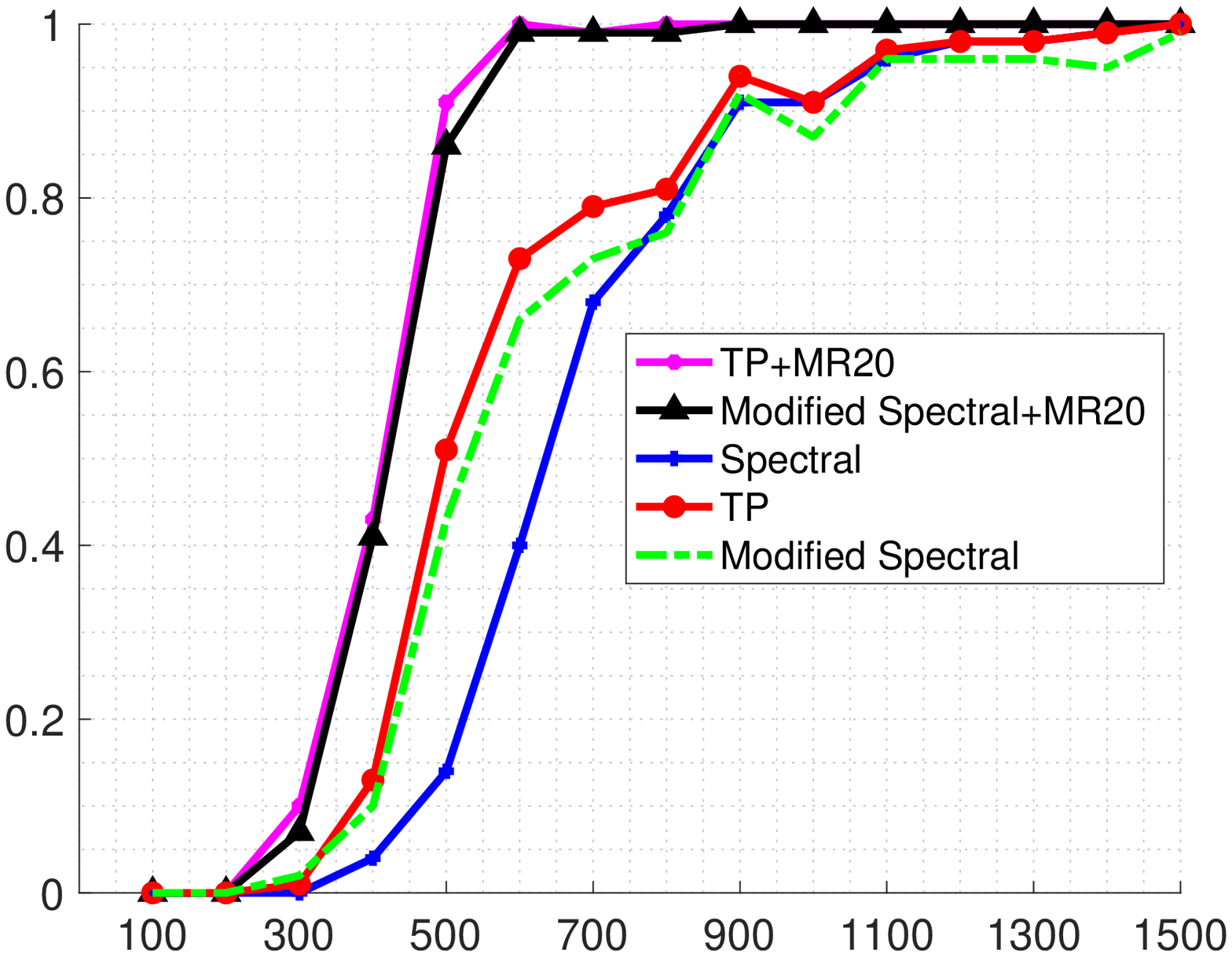}
		\includegraphics[width=0.45\textwidth]{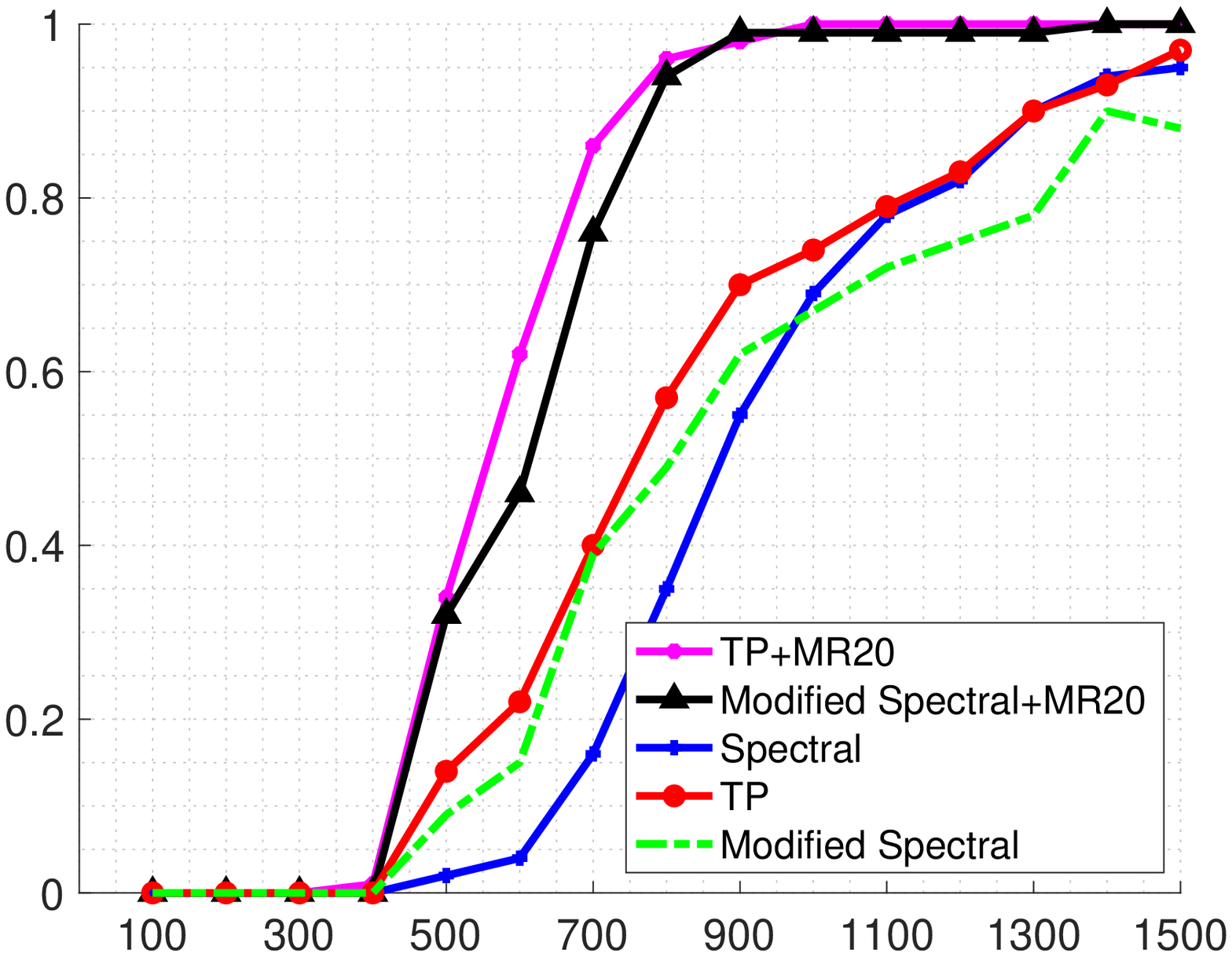}
		\caption{TP, Modified Spectral method, Spectral method, multiple restarted TP ($b=20$), and multiple restarted Modified Spectral method ($b=20$).  Left: Fix $s = 25$; Right: Fix $s = 35$.}
	\end{subfigure}
	\caption{Comparison between TP, Modified Spectral, and Spectral method. In each setting, we conduct 100 independent instances. For TP and Modified Spectral, we can use multiple restarts where we set $b = 20$.}
	\label{fig:fixs}
\end{figure}

\begin{figure}
	\centering
	\begin{subfigure}[t]{.32\linewidth}
		\includegraphics[width=\textwidth]{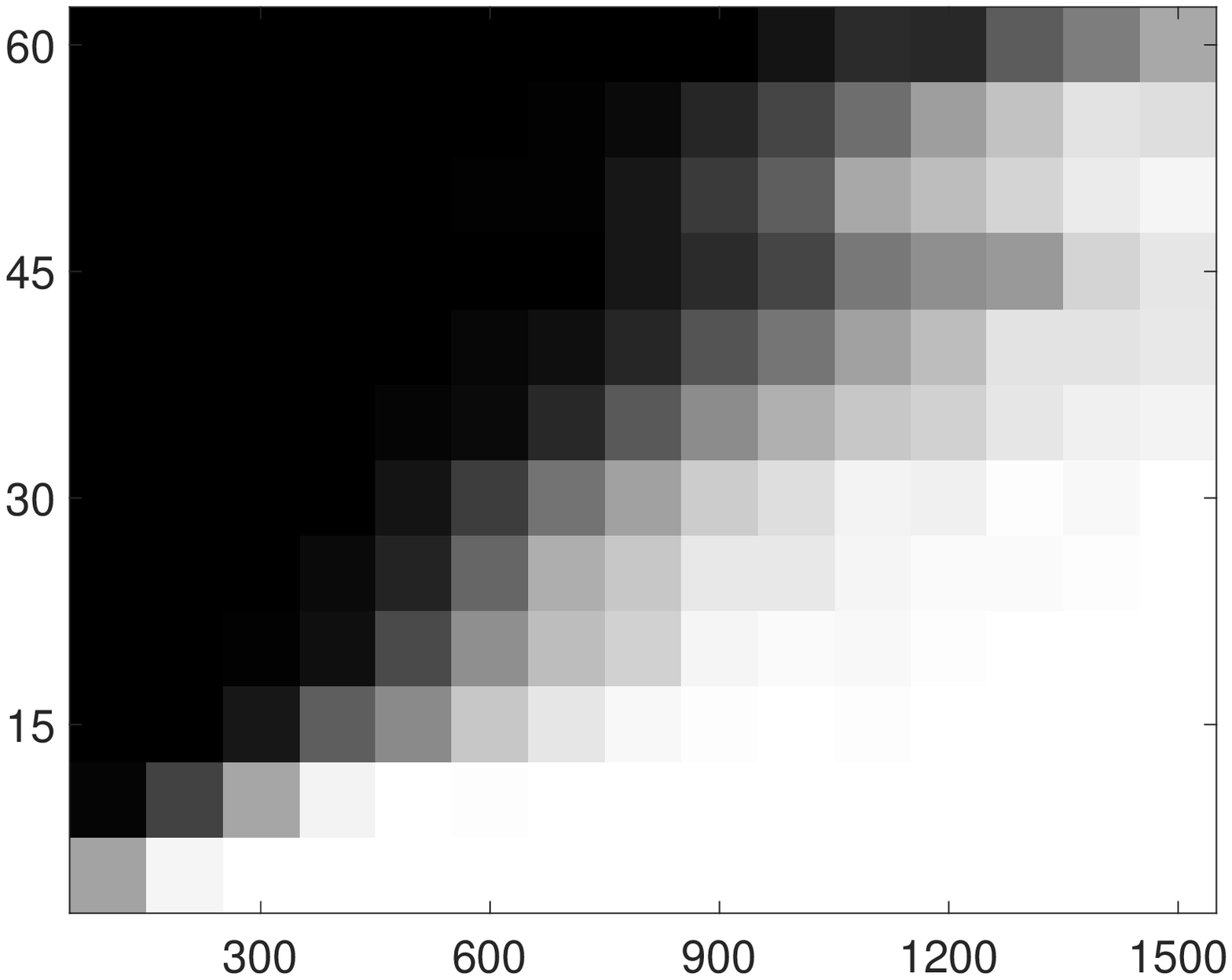}
		\caption{Spectral + HTP}
	\end{subfigure}
	\begin{subfigure}[t]{.32\linewidth}
		\includegraphics[width=\textwidth]{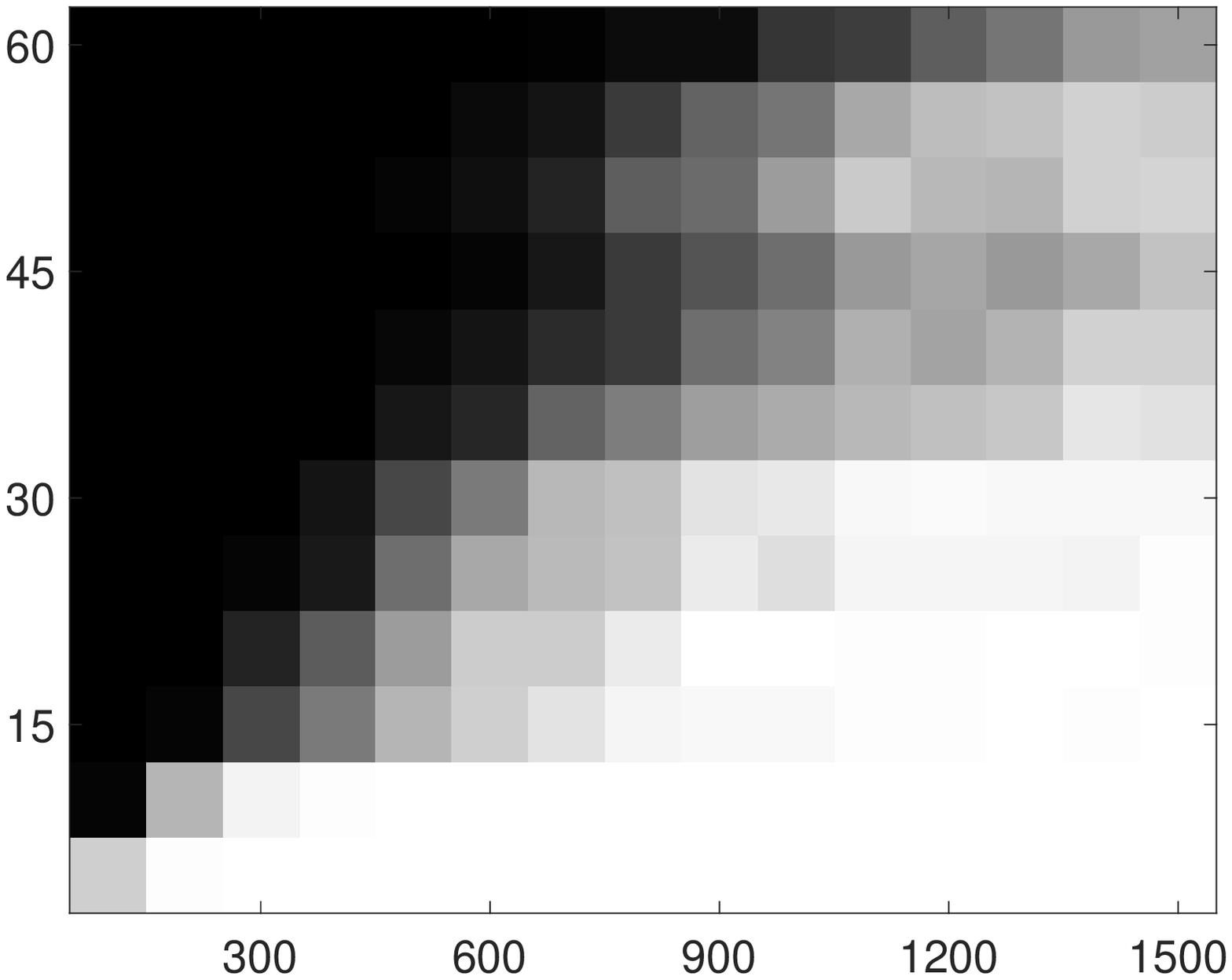}
		\caption{Modified Spectral + HTP}
	\end{subfigure}
	\begin{subfigure}[t]{.32\linewidth}
		\includegraphics[width=\textwidth]{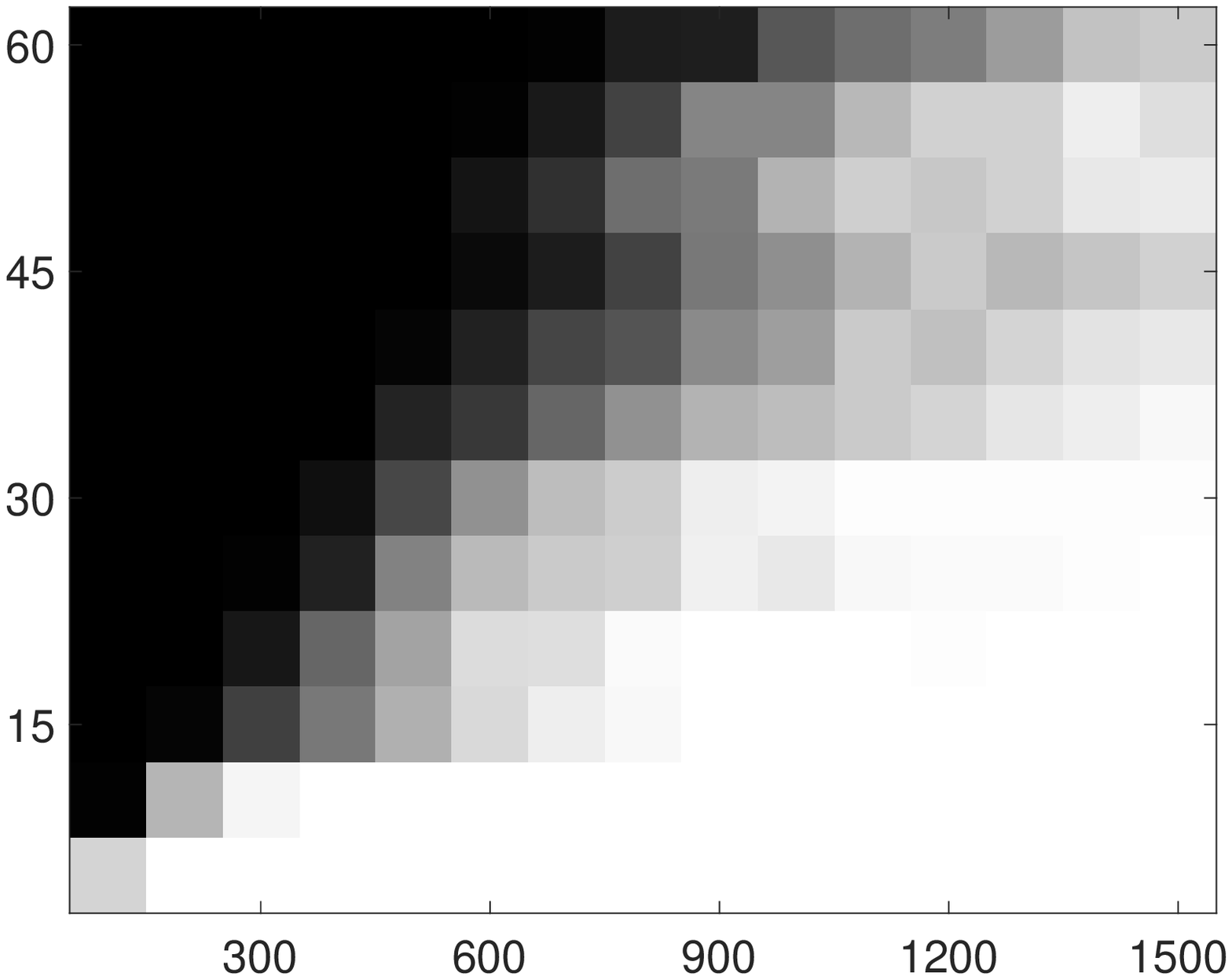}
		\caption{TP + HTP}
	\end{subfigure}
	\caption{Phase transition for different algorithms with signal dimension $n$ = 1000. The successful recovery rates are depicted in different gray levels of the corresponding block. Black means that the successful recovery rate is 0, white 1, and gray between 0 and 1.}
	\label{fig:noMR}
\end{figure}

\begin{figure}
	\centering
	\begin{subfigure}[t]{.32\linewidth}
		\includegraphics[width=\textwidth]{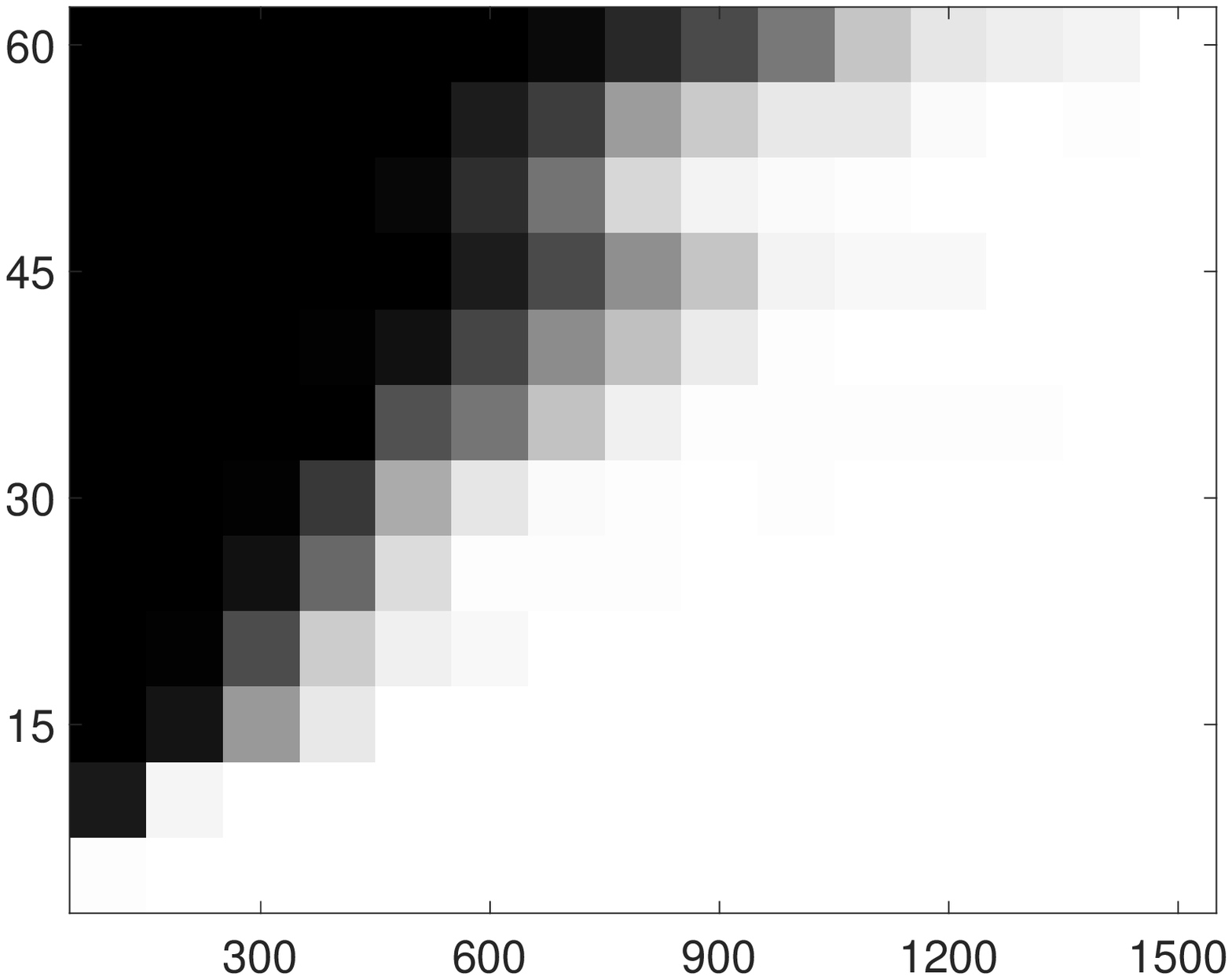}
		\caption{TP + HTP with MR}
	\end{subfigure}
	\begin{subfigure}[t]{.32\linewidth}
		\includegraphics[width=\textwidth]{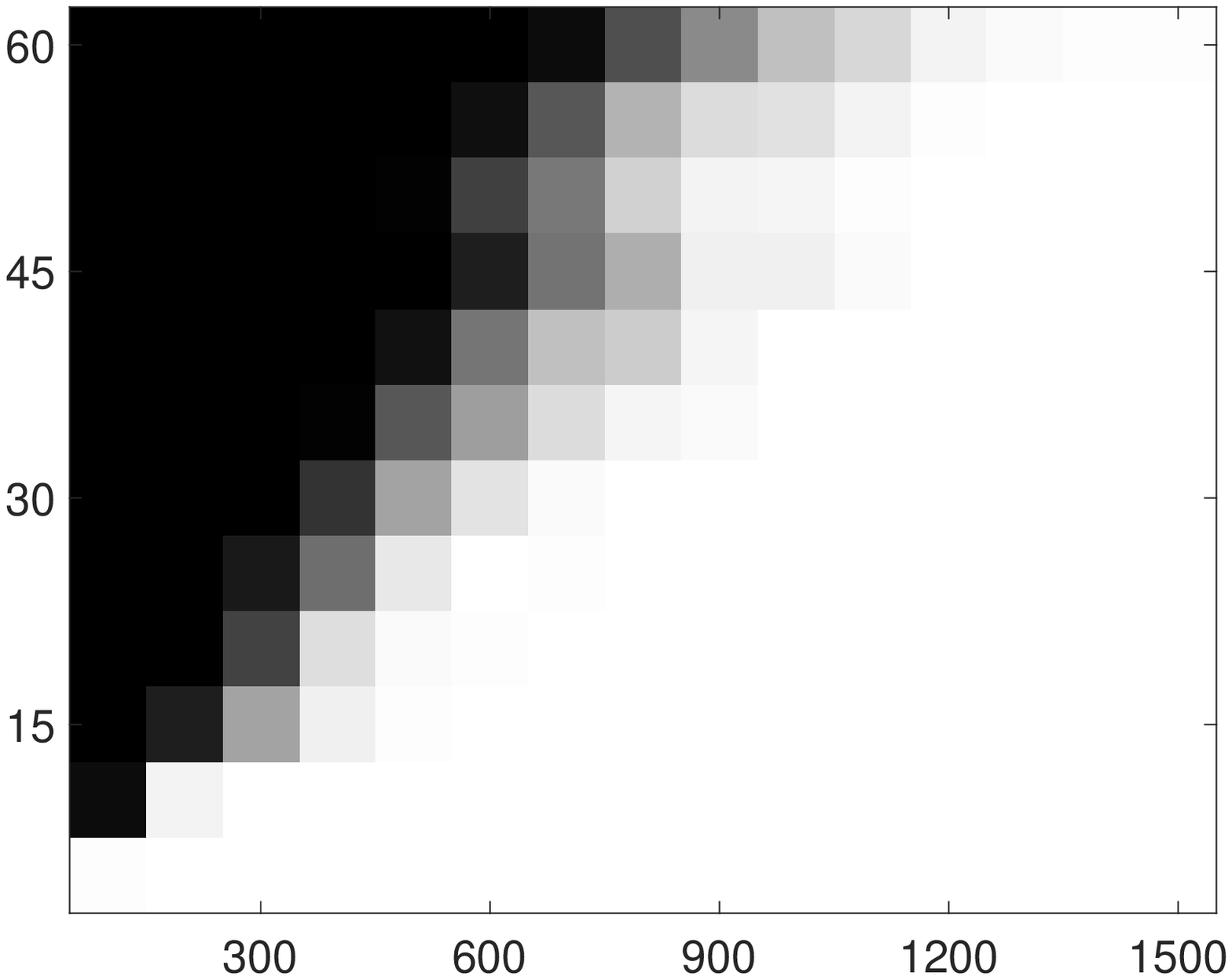}
		\caption{Modified Spectral + HTP with MR}
	\end{subfigure}
	\caption{Phase transition for different multiple-restarted algorithms with signal dimension $n$ = 1000. }
	\label{fig:MR}
\end{figure}

\begin{figure}
	\centering
	\begin{subfigure}[t]{.32\linewidth}
		\includegraphics[width=\textwidth]{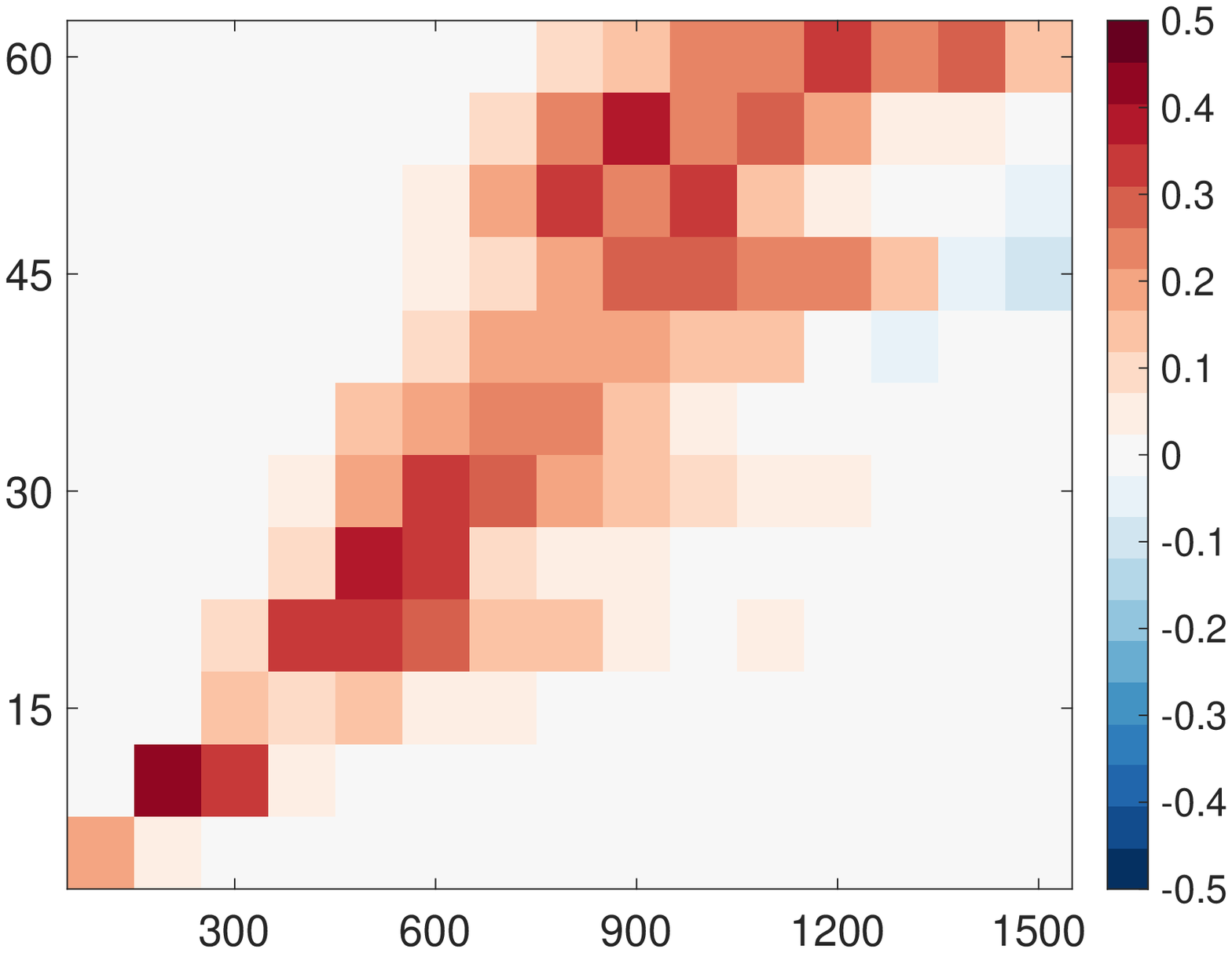}
		\caption{TP and Spectral}
	\end{subfigure}
	\begin{subfigure}[t]{.32\linewidth}
		\includegraphics[width=\textwidth]{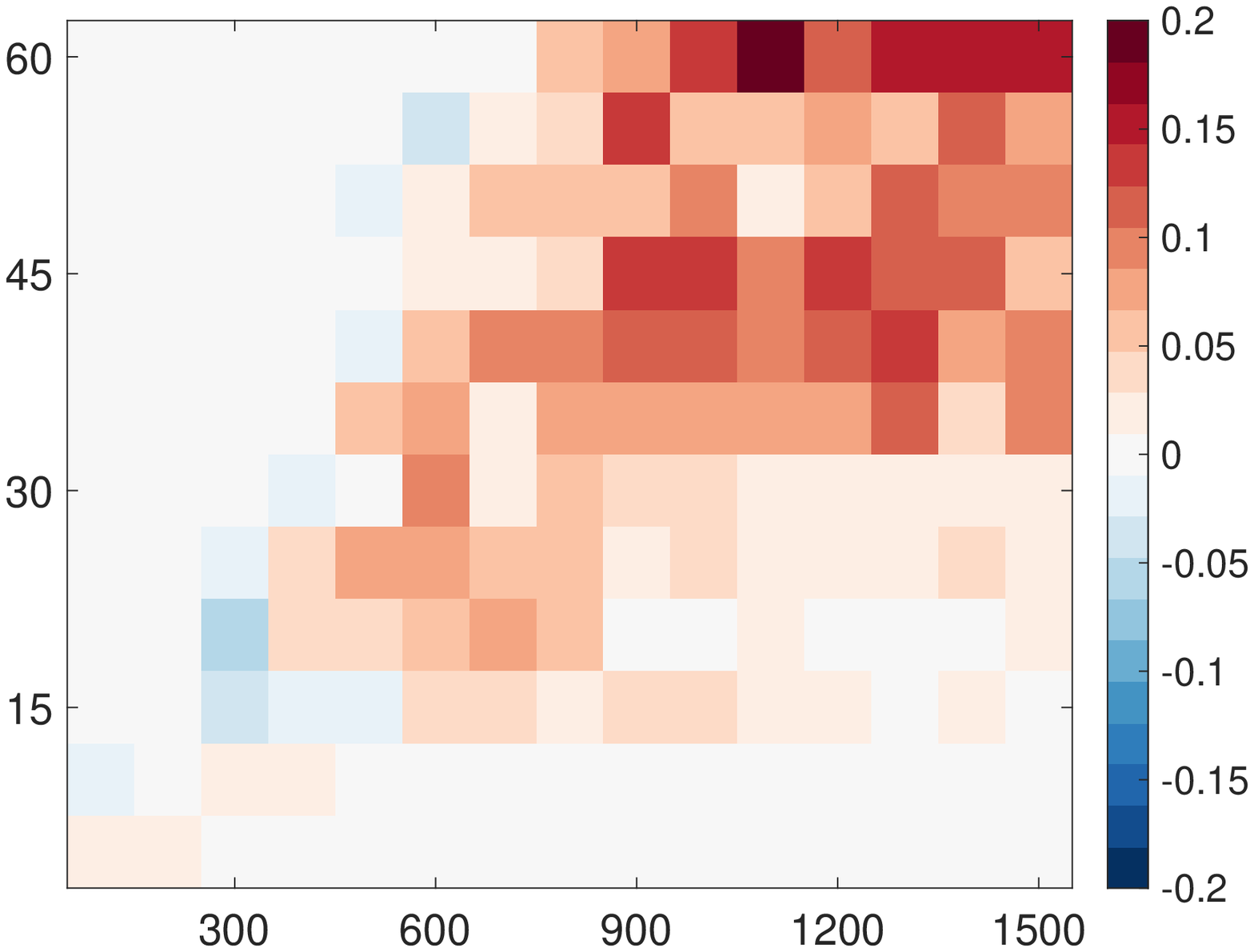}
		\caption{TP and Modified Spectral}
	\end{subfigure}
	\begin{subfigure}[t]{.32\linewidth}
		\includegraphics[width=\textwidth]{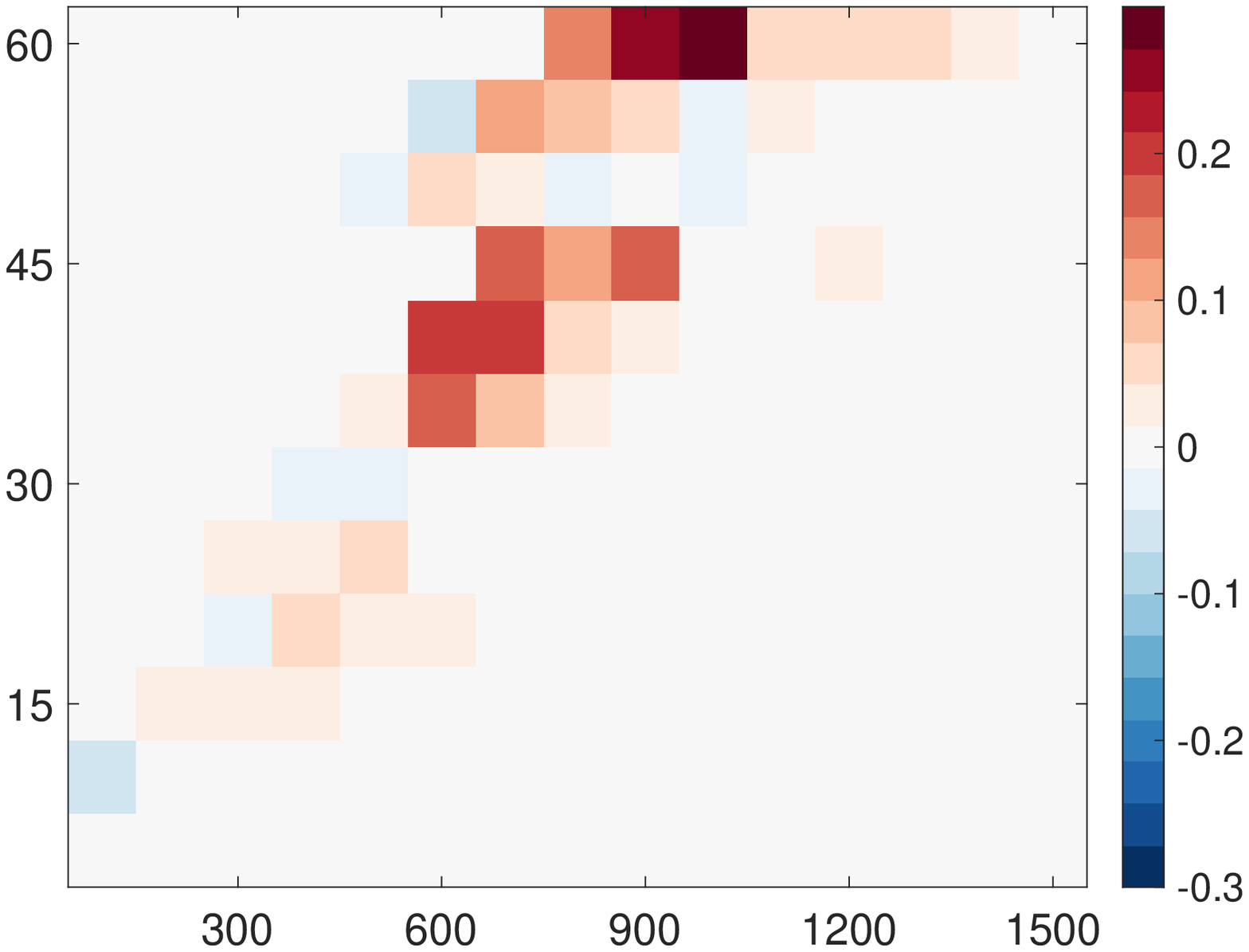}
		\caption{TP and Modified Spectral, both with MR}
	\end{subfigure}
	\caption{Differences between the rate of successful recovery.} %Red means the successful rate of TP is larger than the other method, blue means the opposite. See also the colorbar for the detailed values of difference.}
	\label{fig:diff}
\end{figure}

In the second experiment, we compare the phase transition of different initialization methods. Same as in the previous experiment, to have a fair comparison of different initialization algorithms, we fix the algorithm in the refinement stage to be HTP. We fix the dimension $n = 1000$. We vary the sparsity $s$ from 5 to 60 and the sampling number $m$ from $100$ to $1500$. For TP and Modified Spectral, we also implemented their multiple-restarted versions. The results are shown in \cref{fig:noMR} and \cref{fig:MR}. We see from the figures that TP and Modified Spectral have noticeably higher rates of successful recovery than Spectral.
Moreover, multiple restarts can significantly increase the chance of successful recovery. To illustrate the superiority of TP over the other two methods more clearly, we plot in Figure \ref{fig:diff} the differences in successful recovery rates between TP and Spectral and between TP and Modified Spectral, respectively. We see that TP has a higher successful recovery rate than the other two methods, especially when $m$ is on the margin of the phase transition curve.

\section{Conclusion}

In this work, we have proposed two initialization algorithms for two-stage non-convex sparse phase retrieval approaches. While the popular spectral initialization requires $\Omega(s^2\log n)$ measurements for the desired initialization, the theoretical sample complexity of our proposed algorithms is only $\Omega(\stablesparsity{\x}s\log n)$ unconditionally. Numerical simulations are provided to verify the sample efficiency of the proposed method.  \par

\section*{References}
\bibliographystyle{abbrv}
\bibliography{SparsePRtp}
\end{document}